\documentclass[11pt]{article}
\usepackage{a4,amsmath,amssymb,amsthm,amscd,cite,graphicx}
\usepackage{rotating}
\usepackage[table]{xcolor} 
\usepackage{mathtools, booktabs, caption, float, multirow, comment} 
\usepackage{fixmath, geometry} 
\usepackage[hidelinks]{hyperref} 
\usepackage{tikz}
\usetikzlibrary{intersections,calc,matrix,arrows, decorations.markings,shapes}
\usepackage{multirow}
\newtheorem{theorem}{Theorem}[section]
\usepackage{longtable}
\definecolor{mColor1}{rgb}{0.9,0.9,0.9}  
\definecolor{mColor2}{gray}{0.8}
\definecolor{mColor3}{gray}{1.0}

\newtheorem{prop}[theorem]{Proposition}

\def\Bbb{\mathbb} \def\BZ{\Bbb Z}  \def\BC{\Bbb C}

\newcommand{\be}{\begin{equation}}
\newcommand{\ee}{\end{equation}}
\newcommand{\bea}{\begin{eqnarray}}
\newcommand{\eea}{\end{eqnarray}}
\newcommand{\hl}[1]{{\color{black} #1}}

\catcode`@=11 \@addtoreset{equation}{section} \catcode`@=12

\begin{document}

\begin{titlepage}
\begin{flushright}
 (v. 2) March 2026\\
\texttt{arXiv:2510.24248}
\end{flushright}
\begin{center}
\textsf{\large Quasi-Characters for three-character Rational Conformal Field Theories}\\[12pt]
Suresh Govindarajan$^{\dagger,1}$, Akhila Sadanandan$^{\ddag,1}$ and Jagannath Santara$^{*,2}$ \\[3pt]
Email: $^\dagger$suresh@physics.iitm.ac.in,
$^\ddag$akhila@physics.iitm.ac.in \\ $^*$jagannath.santara@iitgn.ac.in\\[6pt]
${}^1$Department of Physics,\\
Indian Institute of Technology Madras\\
Chennai 600036, India \\[6pt]
${}^2$Department of Physics, \\
Indian Institute of Technology, Gandhinagar,\\ Gujarat 382055, India.
\end{center}
\begin{abstract}
We revisit $(3,0)$ and $(3,3)$ admissible solutions obtained using the MLDE method. We show that  all $(3,0)$  solutions can be written in terms of a universal formula involving the ${}_3F_2$ hypergeometric function that takes into account the monodromy at the elliptic points. We construct $(3,3)$ admissible solutions from $(3,0)$ CFTs using a duality due to Bantay and Gannon. This enables us to compute their modular properties such as the S-matrix and the fusion rules. We find that only 7 of the 15 known $(3,3)$ admissible solutions have proper fusion rules. 

Using the theory of matrix MLDE, starting with a known $(3,0)$ and $(3,3)$ solutions, we construct two other solutions, that are typically quasi-characters that share the same multiplier as the original solution. We then construct linear combinations that lead to new admissible solutions. We observe that admissible solutions arise as integer points that lie on a polytope. We construct all possible $(3,6)$ and $(3,9)$ admissible solutions that arise in this fashion. In some cases, we identify RCFT that arise from our $(3,6)$ admissible solutions. In addition, we obtain a large family of admissible solutions with higher Wronskian index. 
\end{abstract}
\end{titlepage}
\tableofcontents

\section{Introduction}

The holomorphic modular bootstrap aims to classify two-dimensional Rational Conformal Field Theories (RCFTs) -- these are two-dimensional CFTs with a finite number of primaries\cite{Anderson:1987ge}. It can happen that the distinct primaries have the same character. Thus the number of characters of an RCFT (which we denote by $n$) can be less than the number of primaries\footnote{ \hl{For instance, in the case of WZW models, this can happen due to an outer automorphism of the associated Dynkin diagram. In the case of $SO(8)_1$, one has a triality that relates the vector and two spinor characters. This theory has four primaries but only two distinct characters.}}  This approach began with the work of Mathur, Mukhi and Sen (MMS) in the late 80s where they mapped this problem to one of solving Modular Linear Differential Equations(MLDE)\cite{Mathur:1988na,Mathur:1988gt,Naculich:1988xv}. A second parameter, the Wronskian index  (which we denote by $\ell$) \hl{captures the pole structure of the MLDE.} The difficulty of this classification increases as the number of characters as well as the Wronskian index increases. However, when the number of characters, Wronskian index as well as central charge are all small, this problem is somewhat tractable. 

Not unexpectedly, the case of one and two character theories has seen the most progress. For one-character  \hl{bosonic RCFTs}, the central charge is necessarily a multiple of $8$. \hl{This is not true for fermionic RCFTs\cite{Bae:2020xzl}.} There is a unique RCFT at $c=8$, the $E_8$ Kac-Moody Lie algebra at level one, two at $c=16$ and $71$ at $c=24$\cite{Schellekens:1992db}. The numbers increase tremendously when $c>24$\cite{King:2003}. For two-character theories, the Wronskian index is even and non-negative\cite{Naculich:1988xv}. The ones with $\ell=0,2,4$ have been completely classified. More recently, the ones with two primaries, $\ell\geq6$ and $c<25$ have also been classified\cite{Hampapura:2015cea,Mukhi:2022bte}.

The MLDE approach has been used to classify three-character theories. In particular, admissible characters with Wronskian index $\ell=0,3$ have been constructed\cite{Das:2021uvd,Kaidi:2021ent,Bae:2021mej,Gowdigere:2023xnm}. For central charges $c\leq 24$, the unitary RCFTs with vanishing Wronskian index have also been determined\cite{Das:2022uoe}\hl{(see also \cite{Duan:2022ltz})}. However, when the Wronskian index is 3, it is not known whether the admissible characters that have been found are those of RCFTs. 

The method of quasi-characters has been used to construct new admissible characters with Wronskian index $\ell\geq6$\cite{Chandra:2018pjq,Mukhi:2020gnj,Das:2025gto}. Quasi-characters are also solutions to an MLDE for which the $q$-series has integral coefficients that are not necessarily non-negative. For suitable choices of coefficients, adding these quasi-characters to known admissible characters that share the same multiplier system \hl{(i.e., identical $S$ and $T$ matrices)} can lead to admissible solutions, albeit with higher Wronskian index. For instance, linear combinations to two  solutions of $(2,0)$ MLDE's can lead to an admissible character with  Wronskian index $6$. 
  
In another approach, Kaidi et al. use the representation theory of $PSL(2,\mathbb{Z}_N)$ for low values of $N$ and determine possible exponents (modulo one) of all theories with up to 5 characters\cite{Mathur:1989pk,Kaidi:2021ent}. This can be used as an input to the MLDE approach. In particular, the MLDE is uniquely determined by the exponents when the Wronskian index vanishes and the number of  characters is less than six. This needs one to fix the modulo one ambiguity in the Kaidi et al. list which leads to a plethora of choices of exponents even after fixing the Wronskian index.

Given a known RCFT, one can ask whether one can construct new characters from its characters. One such method uses Hecke operators, slightly generalized to act on vector valued modular forms (VVMFs), on the characters of the given RCFT\cite{Harvey:2018rdc,Harvey:2019qzs}. This method has been utilised to construct new admissible characters for a large class of examples with up to six characters\cite{Duan:2022ltz}. This method has the advantage that we know the modular data of the new characters.

Gaberdiel, Hampapura, Mukhi (GHM) propose a coset-like construction that leads to potentially new RCFTs\cite{Gaberdiel:2016zke}. Let
$\chi_i(\tau)$ ($i=0,\ldots,(n-1)$) be the characters of a known RCFT (with central charge $c$ less than 24\footnote{\hl{This constraint on central charge is a practical one as all meromorphic RCFTs up to $c=24$ have been classified\cite{Schellekens:1992db}.}}). Let $m_i$ denote the multiplicities of the characters in the RCFT. Then, the diagonal torus partition function is given by (with $m_0=1$)
\begin{equation}
Z = \sum_{i=0}^{n-1} m_i \ |\chi_i(\tau)|^2\ .
\end{equation}
Let the GHM coset dual have character $\widetilde{\chi}_i(\tau)$. Form the vector consisting  of the characters of given RCFT and its GHM dual. Call them, $\mathbb{X}$ and $\widetilde{\mathbb{X}}$, respectively. Further, let
$M=\text{Diag}(m_0,\ldots,m_{n-1})$ be the matrix of multiplicities. Then, one has the relation
\begin{equation}
\widetilde{\mathbb{X}}^T \cdot M \cdot \mathbb{X} = (J(\tau)+\mathcal{N})\ M\ .
\end{equation}
Here $\mathcal{N}$ is a constant and the coset dual has central charge $(24-c)$. One then `solves' this equation to obtain the $q$-series for $\widetilde{\mathbb{X}}$.

The last approach that is pursued in this paper is to use the theory of vector valued modular forms developed by Bantay and Gannon\cite{Bantay:2007zz,Gannon:2013jua}. The main idea is to start with a known RCFT with VVMF $\mathbb{X}$ and construct a basis of VVMFs with the same multiplier as the known RCFT. Taking linear combinations of the basis with coefficients that are polynomials in $J$ leads to new solutions\cite{Govindarajan:2025rgh}(see also \cite[see appendix C]{Rayhaun:2023pgc}). This is in the spirit of quasi-characters but leads to \textit{all} possible admissible characters that share the same multiplier as the known RCFT. This paper pursues this approach to systematically construct new admissible solutions  with higher Wronskian index starting from $(3,0)$ and $(3,3)$ theories. 

\noindent The main results of this paper are as follows.
\begin{enumerate}
\item We provide a uniform expression for all $(3,0)$ characters and quasi-characters in terms of the ${}_3F_2$
hypergeometric function. While this is not necessarily a new result, it incorporates the monodromy data from the MLDE in a natural way. Proposition \ref{PropHyp} summarises the result.
\item Given a $(3,0)$ RCFT  with VVMF that we denote by $\mathbb{X}$, we construct two additional solutions that we call $\mathbb{Y}_1$ and $\mathbb{Y}_2$ that provide a basis for VVMF with this multiplier. Given the exponents of $\mathbb{X}$ and two \textit{degeneracies}, $y_{21}$ and $y_{31}$, we give an explicit formula, Eq. \eqref{30Xi}, for the solution to the matrix MLDE, Eq. \eqref{MMLDE}, that appears in the theory of Bantay and Gannon.  Multiple tables present explicit results for all known $(3,0)$ RCFTs.
\item Starting with known $(3,0)$ RCFTs with central charge $c<24$, we construct all possible $(3,6)$ admissible characters. Admissible characters appear as integral points inside a convex polytope, that can be non-compact with $(3,6)$ theories appearing at co-dimension one facets and occasionally $(3,0)$ theories appears at corners (co-dimension ridges) of the polytope. Interior points correspond to higher Wronskian index. We show that some of the $(3.6)$ admissible characters are those of an RCFT. Our study is not comprehensive but done only for a  few cases.
\item Bantay and Gannon propose a duality that maps one basis of VVMF with a given multiplier to the basis of a VVMFs with a multiplier derived from the original one. As an application, we show that this duality maps $(3,0)$ theories to $(3,3)$ theories. We obtain all the 15 known solutions that were obtained using the MLDE method in this fashion. This enables us to obtain the S-matrices for the 15 solutions and we find only 7 of them have fusion coefficients that are 0 or 1.
\item Like for the $(3,0)$ examples, we obtain $(3,9)$ admissible solutions starting from the $(3,3)$ admissible solutions.
\end{enumerate}

The organisation of the paper is as follows. Section 1 provides an introduction to the problem being studied as well as a summary of the results. Section 2 provides the necessary background that connects VVMFs and RCFTs. Section 3 focuses on $(3,0)$ RCFTs We first show that all characters and quasi-characters  that arise from $(3,0)$ MLDE have a uniform presentation in terms of the $_3F_2$ hypergeometric function. Then, we extend this presentation to the basis of VVMFs, $\Xi$, associated with any $(3,0)$ RCFT. These bases are presented in multiple tables that provide the information needed to construct the $q$-series for all characters. In section 4, we move to the study of $(3,3)$ theories. In section 4.1, we provide a new construction of all known $(3,3)$ admissible solutions using a duality due to Bantay and Gannon\cite{Bantay:2007zz}. Using this construction, we are able to show that only one class containing $7$ characters of $(3,3)$ theories lead to consistent fusion rules. Section 4.2 obtains quasi-characters and admissible characters for each of these seven theories. In section 5, we look for and find RCFTs among the $(3,6)$ admissible solutions constructed in section 3. We conclude with some remarks. In appendix A, we discuss some of the modular forms that appear in the paper. Appendix B consists of a large number of tables that provide details of the $(3,6)$ theories obtained from $(3,0)$ theories. 

\section{Background}

\subsection{Basic theory of Vector Valued Modular Forms}

Let $\mathbb{X}(\tau)=(\chi_0(\tau),\chi_1(\tau),\ldots,\chi_{n-1}(\tau))^T$ be a Vector Valued Modular Form (VVMF) of rank $n$, weight $w$ and multiplier $\rho$ of $PSL(2,\BZ)$. Then, one has
\[
\mathbf{X}\left(\tfrac{a\tau+b}{c\tau+d}\right)
= (c\tau+d)^w \ \rho\left(\tfrac{a\tau+b}{c\tau+d}\right)\ \mathbb{X}(\tau)\ ,
\]
where $\left(\begin{smallmatrix} a & b \\ c & d \end{smallmatrix}\right)\in PSL(2,\BZ)$. Let $\mathcal{M}^!_w(\rho)$ denote the space of weakly holomorphic VVMFs. Our interest will largely be with VVMFs of weight zero. These appear naturally as the vector of characters of an RCFT with the rank being given by the number of characters.

We will make use of some results from the theory of VVMFs largely due to Bantay and Gannon that are of interest to us\cite{Bantay:2001ni,Bantay:2005vk,Bantay:2007zz,Bantay:2010uy,Gannon:2013jua}. We will assume that the VVMFs of interest are such that the multipliers are admissible so that the theorems we quote are applicable. Consider the following modular differential operators:
\begin{equation}
\nabla_{1,w} := \frac{E_4(\tau)E_6(\tau))}{\Delta(\tau)}\, \mathcal{D}_w\quad,\quad
\nabla_{2,w} := \frac{E_4(\tau)^2}{\Delta(\tau)}\, \mathcal{D}_w^2\quad,\quad\nabla_{3,w} := \frac{E_6(\tau)}{\Delta(\tau)}\, \mathcal{D}_w^3\quad.
\end{equation}
These operators are maps from $\mathcal{M}^!_w(\rho)$ to itself. The following properties are relevant for us.
\begin{enumerate}
\item From \cite[Thm. 3.3(a)]{Gannon:2013jua}, we have $\mathcal{M}^!_w(\rho)$ is a vector space of dimension equal to the rank. Let $\Xi=(\mathbb{X},\mathbb{Y}_1,\ldots,\mathbb{Y}_{n-1})$ be a basis for $\mathcal{M}^!_w(\rho)$.
\item Let $\Xi(\tau)$ denote the $n\times n$ matrix whose columns represent the basis for $\mathcal{M}^!_w(\rho)$. One has the following $q$-expansion that defines $\Lambda$ and $\mathcal{Y}$ as well as fixes the normalisation for the basis.
\begin{equation}\label{vvmfexpansion}
\Xi(\tau) = q^\Lambda \Big( \mathbf{1}_n + \mathcal{Y}\ q + O(q^2)\Big)
\end{equation}
We work typically in a basis where $\Lambda$ is the diagonal matrix with entries $(\lambda_1,\ldots,\lambda_n)$ and $\mathcal{Y}$ is a $n \times n$ matrix whose entries we write as follows:
\begin{equation}
\mathcal{Y}=\begin{pmatrix}
y_{11} & y_{12}& \cdots & y_{1n} \\
y_{21} & \ddots & \ddots & y_{2n} \\
\vdots & \ddots & \ddots & \vdots \\
y_{n1} & y_{n2} & \cdots & y_{nn}
\end{pmatrix}\ ,
\end{equation}
whose entries will generically be rational. However, we will find them to have integral entries in a large class of examples that we consider. We will indicate situations where we obtain rational entries.
\item From  \cite[Thm. 3.3(b)]{Gannon:2013jua} and \cite{Bantay:2007zz}, $\Xi$ is a solution to the following first-order matrix MLDE.
\begin{equation}\label{MMLDE}
\nabla_1 \Xi=  \Xi\left((J-984)\Lambda_w + [\Lambda_w,\mathcal{Y}_w]\right)\ ,
\end{equation}
where $\Lambda_w=\Lambda-\frac{w}{12}\mathbf{1}_n$ and $\mathcal{Y}_w = \mathcal{Y} + 2w \mathbf{1}_n$ and $J=\frac1q+744+\cdots$ is the modular $J$-function.  
\item Given a VVMF $\mathbb{X}$, one obtains $\mathcal{M}^!_w(\rho)$ as a $\BC(\nabla_1,\nabla_2,\nabla_2,J)$ module over $\mathbb{X}$. 
\end{enumerate}

\subsection{From RCFTs to VVMFs}

Consider an $n$ character RCFT with central charge $c$, conformal dimensions, $h_a$ and characters $\chi_a(\tau)$ for $a=0,1,\ldots,(n-1)$. $\chi_0(\tau)$ is taken to be the identity character and hence $h_0=0$. Let $\ell$ denote the Wronskian index of the RCFT.
Form the VVMF $\mathbb{X}$ given by
\begin{equation}
\mathbb{X}=(\chi_0(\tau),\ldots,\chi_{n-1}(\tau))^T\ .
\end{equation}
Define $\alpha_a = -\frac{c}{24}+h_a$. Then, one has the $q$ expansion 
\begin{equation}\label{charexpansion}
\chi_a(\tau) = q^{\alpha_a} \left( m_{0,a} + O(q)\right)\ ,
\end{equation}
with $m_{0,0}=1$. The Wronskian index of the RCFT is given by the formula
\begin{equation}
\ell = \frac{n(n-1)}{2}-6\sum_{a=0}^{n-1}\alpha_a\ .
\end{equation}

We identify the first column of $\Xi$ with $\mathbb{X}$. Then, by comparing Eq. \eqref{vvmfexpansion} and Eq. \eqref{charexpansion}, we see that
\begin{equation}
\lambda_1=\alpha_0\quad,\quad \lambda_{i+1} = (\alpha_i -1)\text{ for } i=1,\ldots,(n-1)\ . 
\end{equation}
\hl{It is important to note that one has to determine the $\lambda_i$  (without mod 1 ambiguity) for the matrix MLDE \ref{MMLDE} to be satisfied.} 
This provides the relation between the exponents that we associate with $\Xi$ with the ones natural to the RCFT.

Finally, using a procedure implicitly described in item 4 of the previous sub-section, we construct $(n-1)$ other VVMF's $\mathbb{Y}_i$ that share the same multiplier as $\mathbb{X}$.

\subsection{Quasi-characters}

In our application, we take $\mathbb{X}$ to be the weight zero VVMF constructed out of the $n$ characters of a known RCFT. The multiplier $\rho$ is fixed by the $S$ and $T$ matrices. 
\begin{equation}
T =\rho\Big(\left(\begin{smallmatrix} 1 & 1 \\ 0 & 1 \end{smallmatrix}\right)\Big)\quad,\quad
S =\rho\Big(\left(\begin{smallmatrix} 0 & -1 \\ 1 & 0 \end{smallmatrix}\right)\Big)\ .
\end{equation}
Repeated applications of the differential operators $\nabla_i$ ($i=1,2,3$) enables one to generate the complete basis. Let $\mathbb{X}$ and $\mathbb{Y}_i$ for $i=1,\ldots,(n-1)$ denote the generators for $\mathcal{M}_w^!(\rho)$. The $\mathbb{Y}_i$ will typically lead to quasi-characters as the coefficients of the $q$-expansion can be negative integers (possibly, rational). For simplicity, we will assume that the $\mathbb{Y}_i$ have integer coefficients.

\subsection{From quasi-characters to admissible ones}

Given an $(n,\ell)$ RCFT with characters forming the VVMF $\mathbb{X}$ and quasi-characters $\mathbb{Y}_i$ that generate the basis of weight zero modular forms with multiplier associated with the RCFT. In the following, we will describe three constructions of candidate VVMF's which might lead to admissible characters. A candidate VVMF is said to be \textit{admissible} if all the coefficients in its $q$-series are non-negative integers. Let $A$ be a subset of $[1,n-1]$ with $|A|=r\leq (n-1)$.
\begin{enumerate}
\item   To the subset $A$, associate the VVMF:
\begin{equation}
\mathbb{U}_{r,A} = \mathbb{X} + \sum_{i\in A} b_i\ \mathbb{Y}_i\ ,
\end{equation}
with $b_i \in \mathbb{Z}_{>0}$. We scan the values of $b_i$ looking for admissible characters.   Looking at the leading exponents of the VVMF $\mathbb{U}_r$  we will generically obtain a $(n,\ell+6r)$ admissible character. However, for special choices of the $b_i$, the Wronskian index can be smaller. For the case when $n=3$, we will have three such VVMF's i.e., $\mathbb{U}_{1,1}$,  
$\mathbb{U}_{1,2}$,  and $\mathbb{U}_{2}$. Of course, $\mathbb{U}_0=\mathbb{X}$.
\item Define
\begin{equation}
\mathbb{W}_{r,A} = (\widehat{J}-y_{11}+b)\ \mathbb{X} - \sum_{i=2}^{n} y_{i1}\ \mathbb{Y}_{i-1}+ \sum_{i\in A} c_{i}\ \mathbb{Y}_{i}\ ,\ ,
\end{equation}
where $b\in \mathbb{Z}_{\geq 0}$ and $c_i\in \mathbb{Z}_{>0}$ and $\widehat{J}=(J-744)$. We search for admissible characters by scanning through the allowed values of the coefficients. For the case, when $n=3$, we will have four classes of VVMFs: $\mathbb{W}_0$, $\mathbb{W}_{1,1}$, $\mathbb{W}_{1,2}$ and $\mathbb{W}_2$.
\item Define
\begin{equation}
\mathbb{V}_{0} = (\widehat{J}+b_1)\Big((\widehat{J}-y_{11})\ \mathbb{X}- \sum_{i=2}^{n} y_{i1} \ \mathbb{Y}_{i-1}\Big) + (b_2^*+ b_2)\ \mathbb{X} +\sum_{i=1}^{n-1} s_i^*\ \mathbb{Y}_{i}\ ,
\end{equation}
where $b_2^*$ and $s_i^*$ are chosen so that the leading terms in $\mathbb{V}_0$ take the form
\begin{equation}
\mathbb{V}_0 = \begin{pmatrix}
q^{\lambda_1-2} \left( 1 + b_1 q + b_2 q^2 + O(q^3)\right) \\[4pt]
q^{\lambda_i} \left (0 + O(q)\right) \text{ for }i=2,\ldots,n
\end{pmatrix}\ .
\end{equation}
Further we need $b_1,b_2 \in \mathbb{Z}_{\geq 0}$. We do not provide explicit formulae for $b_2^*$ and $s_i^*$ as they are messy and unilluminating to see.
From this we see that admissible characters arising in this fashion will generically have Wronskian index $(\ell+12)$. It may happen that for special values of $(b_1,b_2)$, we might find admissible characters with Wronskian index $(\ell+6)$ (in co-dimension one) and admissible characters with Wronskian index $\ell$ at co-dimension zero. The central charge of the potential CFT is $(c+48)$.
Next, define
\begin{equation}
\mathbb{V}_{r,A} = \mathbb{V}_0 + \sum_{i\in A} s_{i}\ \mathbb{Y}_{i}\ ,
\end{equation}
with $s_i\in\mathbb{Z}_{>0}$. If these lead to admissible characters, the Wronskian index will generically be $(\ell+12+6r)$.
\end{enumerate}
Remark: We caution the reader that the notation used here is slightly different from the one used in \cite{Govindarajan:2025rgh}.

\section{$(3,0)$ theories and the hypergeometric ODE}

Let $D=w\frac{d}{dw}$. The ${}_3F_2$ hypergeometric Ordinary Differential Equation (ODE) is given by
\begin{equation}
w (D+\beta_0)(D+\beta_1)(D+\beta_2)\ F = (D-\alpha_0)(D-\alpha_1)(D-\alpha_2)\ F\ .
\end{equation} 
The local exponents in the Frobenius power series expansion about the regular singular points are\cite{Beukers1989}
\begin{equation}\label{localexponents}
\begin{split}
(\alpha_0,\alpha_1,\alpha_2) &\quad\text{ at } w=0\ , \\
(\beta_0,\beta_1,\beta_2) &\quad\text{ at } w=\infty \ ,\\
\left(0,1, 2-\sum_i(\alpha_i+\beta_i)\right) &\quad \text{ at } w=1\ .
\end{split} 
\end{equation}
A basis about $w=0$ is given by (with $(i=0,1,2)$)
\[
w^{\alpha_i}\ {}_3F_2\left(
\alpha_i+\beta_0,\ \alpha_i+\beta_1,\ \alpha_i+\beta_2 ~;~
\alpha_i-\alpha_0+1,\widehat{\alpha_i-\alpha_1+1},\alpha_i-\alpha_2+1~;~w\right)\quad \ ,
\]
where $\widehat{\phantom{abx}}$ indicates that we delete the $(\alpha_i-\alpha_i+1)$ term.

In order to connect the  ${}_3F_2$ hypergeometric ODE with the MLDE for $(3,0)$ RCFTs, we make the following change of variables. Let $w=\frac{1728}{J(\tau)}$. One has the following properties of the $(3,0)$ MLDE.
\begin{enumerate}
\item The cusp $\tau\rightarrow i\infty$ gets mapped to $w=0$ and the elliptic points $\tau= e^{\frac{2\pi i}{3}}$ gets mapped to $w=\infty$ and $\tau=i$ gets mapped to $w=1$. 
\item The monodromy about $w=0$ is given by the $T$-matrix, monodromy about $w=1$ by the $S$ matrix and the monodromy about $w=\infty$ with the $U\hl{(=ST^{-1})}$ matrix.
\item For a unitary RCFT with central charge $c$ and non-zero conformal weights $(h_1,h_2)$, the exponents of the $T$-matrix are 
\[
\alpha_0=-\frac{c}{24}\quad,\quad \alpha_1=-\frac{c}{24}+h_1\quad,\quad\alpha_2=-\frac{c}{24}+h_2\ .
\]
\item For a $(3,\ell)$ RCFT, the Wronskian index $\ell$ is given by
\begin{equation}
\ell = 3-6\sum_{i=0}^{n-1}\, \alpha_i\ . \label{ellformula}
\end{equation}
We see that $\ell=0$ implies that $\sum_i\alpha_i=\frac12$. 
\item Since $U^3=I$, the eigenvalues of $U$ are cube-roots of unity. For $j=0,1,2$, let $b_j$ denote the number of eigenvalues of $U$ with eigenvalue $\exp(\frac{2\pi ij}{3})$. Similarly, $S^2=I$ implies that the eigenvalues of $S$ are $\pm1$. Let $a_1$ denote the  number of eigenvalues of $S$ with eigenvalue $-1$. Then,
\begin{equation}
\ell = -9 + 3 a_1 + 2(b_1+2b_2)
\end{equation}
$\ell=0$ implies only two possible solutions. (i) $a_1=3$, $b_1=b_2=0$ and (ii) $a_1=b_1=b_2=1$. Solution (i) implies that $S=-I$ which would correspond to theories that are tensor product of three one-character RCFTs. We assume that this is not the case. Then, we are left with only solution (ii). In particular this implies that $b_0=b_1=b_2=1$ and $a_0=2$ and $a_1=1$.
\end{enumerate}
\begin{prop} Assuming that the local exponents of the MLDE about $w=0$ do not differ by an integer, the solutions to the $(3,0)$ MLDE are given by solutions to the ${}_3F_2$ hypergeometric ODE. \label{PropHyp}
The three solutions are given by
\begin{equation}\label{30hypsol}
G(\alpha_0,\alpha_1-1,\alpha_2-1)\quad,\quad
G(\alpha_1,\alpha_0-1,\alpha_2-1)\quad,\quad
G(\alpha_2,\alpha_0-1,\alpha_1-1)\ .
\end{equation}
where
\begin{equation*}
    G(a,b,c;w):= \left(\frac{w}{1728}\right)^{a}\  {}_3F_2 \Big(a,a+1/3,a+2/3;  a-b,a-c;w\Big)\ ,
\end{equation*}
\end{prop}
\begin{proof}[Proof]
When the local exponents do not differ by an integer, the monodromy matrix (the $T$-matrix) about $w=0$ is diagonal. If the local exponents differ by an integer, one may have logarithmic solutions that lead to non-diagonal $T$-matrix. We call such cases exceptional and discuss them later.

The $(3,0)$ MLDE in the coordinate $z=\frac{1}{w}$ is given by
\begin{equation}
\Big[\partial_z^3 + \Big(\tfrac2{z}+\tfrac3{2(z-1)}\Big)\partial_z^2  +\Big(\tfrac2{9z^2} + \tfrac4{3z(z-1)} 
+ \tfrac{\mu_1}{z(z-1)} \Big)\partial_z -
\Big(\tfrac{\mu_2}{z^2(z-1)} \Big)\Big]\ \chi=0
\end{equation}
The indicial equation about $z=0$ has exponents $\beta=(1,\tfrac13,\tfrac23)$ and the indicial equation about $z=1$ has exponents $(0,1,\frac12)$. The indicial equation at $z=\infty$ identifies 
\[
\mu_1=(\alpha_0\alpha_1+\alpha_1\alpha_2+\alpha_0\alpha_2)-\tfrac1{18}  \text{ and }\mu_2=-\alpha_0\alpha_1\alpha_2\ .
\]
 This is consistent with the exponents obtained using Eq. \eqref{localexponents}.
\end{proof}
\textbf{Remarks:}
\begin{enumerate}
\item A similar formula can be written for $(2,0)$ RCFTs. The characters are now given in terms of the ${}_2F_1$ hypergeometric functions. Explicitly, one has
\[
\left(\frac{w}{1728}\right)^{\alpha_0}{}_2F_1(\alpha_0,\alpha_0+\tfrac13;\alpha_0-\alpha_1+1;w)\quad,\quad
\left(\frac{w}{1728}\right)^{\alpha_1}{}_2F_1(\alpha_1,\alpha_1+\tfrac13;\alpha_1-\alpha_0+1;w)\ .
\]
Here $a_1=b_1=1$ and $b_2=0$. This agrees with the formula given in \cite{Chandra:2018pjq}.
\item A similar formula can be written for $(2,2)$ RCFTs. Explicitly, one has
\[
\left(\frac{w}{1728}\right)^{\alpha_0}{}_2F_1(\alpha_0,\alpha_0+\tfrac23;\alpha_0-\alpha_1+1;w)\quad,\quad
\left(\frac{w}{1728}\right)^{\alpha_1}{}_2F_1(\alpha_1,\alpha_1+\tfrac23;\alpha_1-\alpha_0+1;w)\ .
\]
Here $a_1=b_2=1$ with $b_1=0$. This agrees with the formula given in \cite{Chandra:2018pjq}.
\item Franc and Mason have shown a large number of $(3,0)$ RCFT's are given by solutions to the ${}_3F_2$ hypergeometric ODE\cite{Franc:2016,Franc:2020}. However, they fixed the exponents $\alpha_i\mod 1$ in their study. To the extent that we checked, their formulae match ours.
\item All solutions that appeared in the $(3,0)$ modular bootstrap fit the given formula\cite{Das:2021uvd}. We will discuss the case when the exponents about $w=0$ differ by an integer separately.
\item The infinite families of quasi-characters constructed by Mukhi, Poddar and Singh also fit this formula\cite{Mukhi:2020gnj}.
\end{enumerate}

\subsection{The exceptional cases}

We consider the examples obtained from the $(3,0)$ modular bootstrap\cite{Das:2021uvd}. There are four families of examples, two with $c=8$ and another two with $c=16$ that appear in the table.

\begin{enumerate}
\item First, consider the two examples of $(c=8,h_1=h_2=3/4)$ and $(c=16,h_1=h_2=5/4)$. Here the two of the exponents are equal, $\alpha_1=\alpha_2$. The Frobenius power series method gives two solutions that are given by the hypergeometric functions given in the conjecture. They are
\begin{equation}
\begin{split}
F_0&=\left(\frac{w}{1728}\right)^{\alpha_0}{}_3F_2(\alpha_0, \alpha_0+\tfrac13,\alpha_0+\tfrac23;\alpha_0-\alpha_1+1;w)\ ,\\
F_1&=\left(\frac{w}{1728}\right)^{\alpha_1}{}_3F_2(\alpha_1,\alpha_1+\tfrac13,\alpha_1+\tfrac23;\alpha_1-\alpha_0+1;w)\ .
\end{split}
\end{equation}
The third solution is given by
\[
F_2 = F_1 \log w + \text{power series in } q\ .
\]
The solution $F_1$ turns out to be unstable in both cases -- the coefficients are rational but there is no single factor that makes them integral to all orders.
\item Next, consider the case of $(c=8,h_1=1/2,h_2=1)$. Solving the MLDE leads to an infinite family of solutions (parametrised by $b$) given by
\begin{equation}
    \mathbb{X} = \begin{pmatrix}
        G(-\tfrac13,-\tfrac56,-\tfrac13;w) + b\  G(\tfrac23,-\tfrac13,-\tfrac45;w) \\
        G(\tfrac16,-\tfrac45,-\tfrac13;w) \\
        G(\tfrac23,-\tfrac13,-\tfrac45;w)
    \end{pmatrix}\ ,
\end{equation}
Positivity of the $q$-series requires $b\geq -248$. The RCFT $D_{4,1}^{\otimes2}$ appears when $b=-192$ and the RCFT $D_{8,1}$ apears when $b=-128$.
\item Next, consider the case of $(c=16,h_1=1,h_2=3/2)$. Solving the MLDE leads to an infinite family of solutions (parametrised by $b$) given by
\begin{equation}
    \mathbb{X} = \begin{pmatrix}
        G(-\tfrac23,-\tfrac23,-\tfrac16;w) + b\  G(\tfrac13,-\tfrac53,-\tfrac16;w) \\
        G(\tfrac13,-\tfrac53,-\tfrac16;w) \\
        G(\tfrac56,-\tfrac53,-\tfrac23;w)
    \end{pmatrix}\ ,
\end{equation}
Positivity of the $q$-series requires $b\geq -496$. The RCFT given by $\mathcal{E}_1(D_{4,1}^{\otimes4})$, which is a one-character extension of $D_{4,1}^{\otimes4}$,   appears when $b=-384$.

 \item Next consider the case of $(c=16,h_1=1/2,h_2=2)$. Solving the MLDE leads to an infinite family of solutions (parametrised by $b$) given by
\begin{equation}
    \mathbb{X} = \begin{pmatrix}
        G(-\tfrac23,-\tfrac76,\tfrac13;w) + b\   G(\tfrac43,-\tfrac53,-\tfrac76;w) \\
        G(\tfrac16,-\tfrac53,\tfrac13;w) \\
        G(\tfrac43,-\tfrac53,-\tfrac76;w)
    \end{pmatrix}\ ,
\end{equation}
Positivity of the $q$-series requires $b\geq -69752$. The RCFT associated with $D_{16,1}$ appears when $b=-32768$. 
\end{enumerate}

\subsection{$(3,0)$ Quasi-characters from VVMFs}

Define $\lambda_1=\alpha_0$ and $\lambda_{i+1}=\alpha_i-1$ for $i=1,2$.
These are the exponents that appear in the Matrix MLDE associated with a $(3,0)$ admissible character. The VVMF associated with an RCFT with central charge $c$ and non-zero conformal weights $h_1$ and $h_2$ is given by
\begin{equation}
    \mathbb{X} = \begin{pmatrix}
        G(\lambda_1,\lambda_2,\lambda_3;z) \\
       y_{21}\ G(\lambda_2+1,\lambda_1-1,\lambda_3;z) \\
        y_{31}\ G(\lambda_3+1,\lambda_1-1,\lambda_2;z)
    \end{pmatrix}\ ,
\end{equation}
where $\lambda_1=-\frac{c}{24}$, $\lambda_2=-\frac{c}{24}+h_1-1$ and $\lambda_3=-\frac{c}{24}+h_2-1$. The degeneracies $y_{21}$ and $y_{31}$ are determined by details of the RCFT when they are known. In other cases, we can choose the smallest values such that the coefficients of the $q$-series are integral.

The matrix of solutions is then are determined by the data 
\begin{equation}
\mathcal{Y} =\begin{pmatrix}
 y_{11} & y_{12} & y_{13} \\
 y_{21} & y_{22} & y_{23} \\
 y_{31} & y_{32} & y_{33} 
\end{pmatrix}\ ,
\end{equation}
and
\begin{equation}\label{30Xi}
{\footnotesize
\Xi(\tau)=\begin{pmatrix}
G\left(\lambda_1,\lambda_2,\lambda_3; z(\tau)\right) &
y_{21} G\left(\lambda_1+1,\lambda_2-1,\lambda_3; z(\tau)\right) &
y_{31} G\left(\lambda_1+1,\lambda_2,\lambda_3-1; z(\tau)\right)\\
y_{12} G\left(\lambda_2+1,\lambda_1-1,\lambda_3; z(\tau)\right) &
G\left(\lambda_2,\lambda_1,\lambda_3; z(\tau)\right) &
y_{32} G\left(\lambda_2+1,\lambda_1,\lambda_3-1; z(\tau)\right) \\
 y_{13} G\left(\lambda_3+1,\lambda_1-1,\lambda_2; z(\tau)\right) &
 y_{23} G\left(\lambda_3+1,\lambda_1,\lambda_2-1; z(\tau)\right) &
 G\left(\lambda_3,\lambda_1,\lambda_2; z(\tau)\right) 
\end{pmatrix}
}\ ,
\end{equation}

\begin{equation}
\left(\nabla_1 - (\alpha_0 J+b_1)\right)\, \mathbb{X} 
= y_{21} (\alpha_1-\alpha_0) \mathbb{Y}_1 + y_{31} (\alpha_2-\alpha_0) \mathbb{Y}_2\ ,
\end{equation}
where
\[
b_1 = \frac{192 \alpha _0 \left(-9 \alpha _1 \left(\alpha _2-1\right)+9
   \alpha _2+9 \alpha _0 \left(\alpha _1+\alpha
   _2-1\right)-7\right)}{\left(\alpha _0-\alpha _1+1\right)
   \left(\alpha _0-\alpha _2+1\right)}\ .
\]
\begin{equation}
\left(\nabla_2 - (a_0 J+b_2)\right)\, \mathbb{X} 
= y_{21} (a_1-a_0) \mathbb{Y}_1 + y_{31} (a_2-a_0) \mathbb{Y}_2\ ,
\end{equation}
where $a_i=(\alpha _i^2-\tfrac{\alpha _i}{6})$ for $i=0,1,2$ and
\[
b_2=\frac{32 \alpha _0 \left(3 \alpha _0+1\right) \left(18 \alpha
   _0^2+3 \left(6 \alpha _1+6 \alpha _2+1\right) \alpha _0-18
   \alpha _1 \left(\alpha _2-1\right)+18 \alpha
   _2-8\right)}{\left(\alpha _0-\alpha _1+1\right) \left(\alpha
   _0-\alpha _2+1\right)}\ .
\]

\noindent The entries of the $\mathcal{Y}$ are determined in terms of the exponents $\alpha_i$ and $y_{21}$ and $y_{31}$. 
\begin{align*}
y_{ii}&=\frac{1728\, \lambda_i \left(\lambda_i+\frac{1}{3}\right)
   \left(\lambda_i+\frac{2}{3}\right)}{\prod_{j\neq i}(\lambda_i-\lambda_j) }-744\, \lambda_i\\ 
\frac{y_{23}y_{31}}{y_{21}} & = \frac{192 \left(\alpha _0-\alpha _1\right) \left(9 \alpha _2^3-18
   \alpha _2^2+11 \alpha _2-2\right)}{\left(\alpha _0-\alpha
   _2\right) \left(\alpha _0-\alpha _2+1\right) \left(-\alpha
   _1+\alpha _2-1\right) \left(\alpha _2-\alpha _1\right)} \\
  \frac{y_{32}y_{21}}{y_{31}} &= \frac{192 \left(\alpha _0-\alpha _2\right)\left(9 \alpha _1^3-18 \alpha _1^2+11 \alpha
   _1-2\right) }{\left(\alpha
   _0-\alpha _1\right) \left(\alpha _0-\alpha _1+1\right)
   \left(\alpha _1-\alpha _2-1\right) \left(\alpha _1-\alpha
   _2\right)}\\
   y_{12}y_{21}&={\scriptstyle \frac{18432 \alpha _0 \left(9 \alpha _0^2+9 \alpha _0+2\right)
   \left(18 \alpha _0^3-9 \left(2 \alpha _1-2 \alpha _2-5\right)
   \alpha _0^2-9 \left(2 \alpha _1^2+\left(4 \alpha _2-2\right)
   \alpha _1-6 \alpha _2-1\right) \alpha _0+\left(\alpha
   _1-1\right) \left(-36 \alpha _2+9 \alpha _1 \left(2 \alpha
   _2-3\right)+14\right)\right)}{\left(\alpha _0-\alpha _1\right)
   \left(\alpha _0-\alpha _1+1\right){}^2 \left(\alpha _0-\alpha
   _1+2\right) \left(\alpha _0-\alpha _2+1\right) \left(\alpha
   _1-\alpha _2\right)}}\\
y_{13}y_{31} & = {\scriptstyle -\frac{18432 \alpha _0 \left(9 \alpha _0^2+9 \alpha _0+2\right)
   \left(18 \alpha _0^3+9 \left(2 \alpha _1-2 \alpha _2+5\right)
   \alpha _0^2-9 \left(2 \alpha _2^2-2 \alpha _2+\alpha _1
   \left(4 \alpha _2-6\right)-1\right) \alpha _0+\left(18 \alpha
   _1 \left(\alpha _2-2\right)-27 \alpha _2+14\right)
   \left(\alpha _2-1\right)\right)}{\left(\alpha _0-\alpha
   _1+1\right) \left(\alpha _0-\alpha _2\right) \left(\alpha
   _0-\alpha _2+1\right){}^2 \left(\alpha _0-\alpha _2+2\right)
   \left(\alpha _1-\alpha _2\right)}}
\end{align*}
\textbf{Remarks:} 1. \hl{The degeneracies} $y_{21}$ and $y_{31}$ are not fixed here and have to be given as input. \\ 
2. One has the relation $y_{11}+y_{22}+y_{33}=252$.

\subsection{Examples}

In this subsection, we will present the basis of quasi-characters that arise from $(3,0)$ RCFTs that have been discussed in \cite{Das:2022uoe} (for the unitary examples) and \cite{Duan:2022ltz} (for the non-unitary examples). The results are presented mostly as a series of tables and organised by RCFTs that share the same S-matrix and potentially belong to the same Modular Tensor Category (MTC). The tables further give the multiplicities $m_1$ and $m_{2}$ for each type. Each row gives the following data: $(c,h_1,h_2)$ and $\mathcal{Y}$ matrix. This data is sufficient to write out the associated $\Xi$ using Eq. \eqref{30Xi}.

Interestingly, $(3,6)$ admissible characters of type $W_0$ arise for $b>0$ in all $(3,0)$ examples. So we do not include this in our tables.

\subsubsection{Quasi-characters for the $B_r$ series}

The $B_{r,1}$ ($r\geq 0$) series with $c=\frac{2r+1}{2}$, $h_1=\frac{1}{2}$, $h_2=\frac{2r+1}{16}$. For these theories, one has $y_{21}=2r+1$ and $y_{31}=2^r$ leading to the following $\mathcal{Y}$-matrix:
\begin{equation}
\mathcal{Y}= \begin{pmatrix}
 r (2 r+1) & -\frac{1}{3} (2 r-31) \left(4 r^2+68 r+225\right) &
   2^{12-r} (2 r-23) \\
 2 r+1 & -(r-11) (2 r+25) & -2^{12-r} \\
 2^r & -2^r (2 r+25) & 2 r-23 \\
\end{pmatrix}
\end{equation}
Here $r=0$ corresponds to the Ising model, $r=1$ is also $A_{1,2}$ and $r=2$ is also $C_{2,1}$. Note that column three of the $\mathcal{Y}$ matrix is fractional for $r>12$. Multiplying $\mathbb{Y}_2$ by a suitable power of $2$ ($2^{r-12}$), we obtain quasi-characters with integral entries. Alternately, when we construct admissible characters by taking linear combinations involving $\mathbb{Y}_2$, we take the coefficients to be integer multiples of the same power of two.

All theories share the same $S$-matrix:
\[
S=\begin{pmatrix}\frac{1}{2} &\frac{1}{2} &\frac{1}{\sqrt{2}}\\
    \frac{1}{2} &\frac{1}{2} & -\frac{1}{\sqrt{2}}\\
     \frac{1}{\sqrt{2}} &-\frac{1}{\sqrt{2}} & 0
     \end{pmatrix}\ .
\]

\subsubsection{Quasi-characters for the $D_r$ series}

For the $D_{r,1}$ series, one has $c=r$, $h_1=\frac{1}{2}$, $h_2=\frac{r}{8}$. For these theories, one has $y_{21}=2r$ and $y_{31}=2^{r-1}$ leading to the following $\mathcal{Y}$-matrix:
\begin{equation}
\mathcal{Y}=\begin{pmatrix}
r (2 r-1) & -\frac{8}{3} (r-16) \left(r^2+16 r+48\right) &
   2^{14-r} (r-12) \\
 2 r & -(r+12) (2 r-23) & -2^{13-r} \\
 2^{r-1} & -2^r (r+12) & 2 (r-12) \\
\end{pmatrix}
\end{equation}
Here $D_{2,1}$ is also the same as $A_{1,1}^{\otimes2}$, and $D_{3,1}$ is the same as $A_{3,1}$. For $D_4$, we end up with a two-character theory as $h_1=h_2=\tfrac12$.  Note that column three of the $\mathcal{Y}$ matrix is fractional for $r>13$. Multiplying $\mathbb{Y}_2$ by a suitable power of $2$, we obtain quasi-characters with integral entries. Alternately, when we construct admissible characters by taking linear combinations involving $\mathbb{Y}_2$, we take the coefficients to be integer multiples of the same power of two.

The character corresponding to $h_2=\frac{r}{8}$ has multiplicity of two in the RCFT and thus this theory has four primaries. The unresolved and resolved $S$-matrix are
\begin{equation}
{\scriptstyle
S_{\text{unresolved}}= \frac{1}{2}\begin{pmatrix}
1 & 1 & 2 \\
1 & 1 & -2  \\
1 & -1 & 0 
\end{pmatrix}
\quad,\quad
S_\text{resolved}= \frac{1}{2}\begin{pmatrix}
1 & 1 & 1 & 1\\
1 & 1 & -1 & -1 \\
1 & -1 & i^{-r} & -i^{-r} \\
1 & -1 & -i^{-r} & i^{-r}
\end{pmatrix}\ .
}
\end{equation}
\hl{The unresolved $S$-matrix is associated with the distinct characters of the RCFT while the resolved $S$-matrix is associated with primaries of the RCFT.}
\subsubsection{Quasi-characters for other unitary examples}

\begin{longtable}{c||ccc|c}  
\toprule
S. No.  & $c$& $h_1$ &$h_2$&$\mathcal{Y}$ matrix\\  
\endfirsthead
\midrule
\rowcolor{mColor2}  I. & \multicolumn{4}{c}{GHM Dual to $B_{r,1}$ models ($0\leq r\leq15$)} \\ 
& \multicolumn{4}{c}{Multiplicities: $m_1=1,\ m_2=$1}\\
    & \multicolumn{4}{c}{S matrix: 
    $\begin{pmatrix}\frac{1}{2} &\frac{1}{2} &\frac{1}{\sqrt{2}}\\
    \frac{1}{2} &\frac{1}{2} & -\frac{1}{\sqrt{2}}\\
     \frac{1}{\sqrt{2}} &-\frac{1}{\sqrt{2}} & 0
     \end{pmatrix}$}  \\ \hline

    $\ell=0$  & $\frac{47-2r}{2}$&$\frac{3}{2}$ &$\frac{31-2r}{16}$  & $\begin{pmatrix}
        r(47-2r)& 23-2r & 2^{1+r}\\
       \frac{1}{3}(47-2r)(31-2r)(9+2r) & (11-r)(23-2r) &-2^{1+r}(1+2r) \\
        2^{11-r}(47-2r)& -2^{11-r} & -(1+2r) 
    \end{pmatrix}$\\
    \midrule
    \midrule
    
\rowcolor{mColor2} II. & \multicolumn{4}{c}{GHM Dual to $D_{r,1}$ models ($0\leq r\leq 15$)} \\ 
& \multicolumn{4}{c}{Multiplicities: $m_1=1,\ m_2=2$}\\
    & \multicolumn{4}{c}{S matrix: 
    $\begin{pmatrix}
    
    \frac{1}{2} &\frac{1}{2} & 1 \\
        \frac{1}{2}  &\frac{1}{2} & -1\\
        \frac{1}{2} & -\frac{1}{2} & 0 \\
        
    \end{pmatrix}$ }  \\ \hline

       $\ell$=0  & $24-r$&$\frac{3}{2}$ &$\frac{16-r}{8}$  & $\begin{pmatrix}
       
       - (24-r)(2r-1) & 24-2r  & 2^{1+r} \\
       \frac{8}{3}(24-r)(16-r)(4+r)& (12-r)(23-2r)&-2^{1+r}r  \\
        2^{12-r}(24-r)& -2^{11-r} & -2r
       
    \end{pmatrix}$\\

    \bottomrule
\caption{$(3,0)$ models that appear as GHM coset duals of the $B_{r,1}$ and $D_{r,1}$ theories.}
\end{longtable}

\begin{longtable}{c||ccc|c}  
\toprule
S. No.  & $c$& $h_1$ &$h_2$&$\mathcal{Y}$ matrix\\  
\endfirsthead
\midrule
\rowcolor{mColor2} III. & \multicolumn{4}{c}{Type $A_{2,1}^{\otimes 2}$} \\ 
& \multicolumn{4}{c}{Multiplicities: $m_1=4,\ m_2=4$}\\
    & \multicolumn{4}{c}{S matrix: $ \begin{pmatrix}
    
 \frac{1}{3}  &  \frac{4}{3} & \frac{4}{3} \\
 \frac{1}{3} & \frac{1}{3} & -\frac{2}{3} \\
 \frac{1}{3} & -\frac{2}{3} & \frac{1}{3} \\

\end{pmatrix}$}\\
\cmidrule{2-5}

\rowcolor{mColor1} & $4$  &  $\frac13$  & $ \frac23 $ & $\begin{pmatrix}

 16 & 34992 & 2916 \\
 3 & 80 & -54 \\
 9 & -2430 & 156 \\

\end{pmatrix}$\\
& 20& $\frac 43$ &$\frac53$ & $\left(
\begin{array}{ccc}
 80 & 108 & 12 \\
 1215 & 156 & -18 \\
 8748 & -1458 & 16 \\
\end{array}
\right)$\\
\rowcolor{mColor1} & 12& $\frac 23$ &$\frac43$ &$\left(
\begin{array}{ccc}
 156 & 4860 & 36 \\
 27 & 80 & -6 \\
 729 & -17496 & 16 \\
\end{array}
\right)$\\
\midrule
\rowcolor{mColor2} IV. & \multicolumn{4}{c}{Type $A_{4,1}$} \\ 
& \multicolumn{4}{c}{Multiplicities: $m_1=2,\ m_2=2$}\\
    & \multicolumn{4}{c}{S matrix: $\frac{1}{\sqrt{5}} \begin{pmatrix}
     1 & 2 & 2\\
1& \frac{1}{2} \left(\sqrt{5}-1\right)  & -\frac{1}{2} \left(\sqrt{5}+1\right) \\
1&- \frac{1}{2} \left(\sqrt{5}+1\right) & \frac{1}{2} \left(\sqrt{5}-1\right) \\
\end{pmatrix}$}\\
\cmidrule{2-5}
\rowcolor{mColor1} & $4$ & $\frac25$ & $\frac35$ & $\begin{pmatrix}

 24 & 14950 & 3400 \\
 5 & 92 & -119 \\
 10 & -1196 & 136 \\

\end{pmatrix}$\\
&20& $\frac75$ & $\frac85$ & $\left(
\begin{array}{ccc}
 120 & 52 & 14 \\
 2500 & 104 & -49 \\
 8125 & -676 & 28 \\
\end{array}
\right)$\\

\bottomrule
\caption{$(3,0)$ models of type $A_{2,1}^{\otimes2}$ and type $A_{4,1}$}
\end{longtable}

\subsubsection{Quasi-characters for some non-unitary examples}

For non-unitary theories, the quantities $(c,h_1,h_2)$ refer to the unitary presentation. By this, $c$ refers to what is sometimes called $c_{\text{eff}}$ which equals $(c-24 h_{\text{min}})$, where $c$ is the true cental charge and $h_{\text{min}}$ is the minimal conformal weight. Similarly, $(h_1,h_2)$ are shifted upwards by $-h_{\text{min}}$. Note that $h_{\text{min}}<0$ non-unitary theories and $h_{\text{min}}=0$ for unitary theories. \hl{It is important to note that one uses the true identity operator in the Verlinde formula to get well-defined fusion rules. This is discussed for the non-unitary minimal model, $M(5,2)$,  in \cite{Mathur:1988gt}.}

Type $LY_2$ refers to the class of non-unitary RCFTs in the class of the minimal model $M(7,2)$ and $LY_1$ refers to the minimal model, $M(5,2)$ associated with the Lee-Yang model\cite{DiFrancesco:1997nk,Duan:2022ltz}. The $S$-matrices quoted have been obtained using the results of Mukhi et al.\cite{Mukhi:2019cpu}.

\begin{longtable}{c||ccc|c|c}

S. No.  & $c$& $h_1$ &$h_2$&$\mathcal{Y}$ matrix &S matrix\\

\rowcolor{mColor2} V.& \multicolumn{5}{c}{Type $LY_2$}\\ 
& \multicolumn{5}{l}{Multiplicities: $m_1=1,\ m_2=1$}\\
    & \multicolumn{2}{l}{ S matrix: }& \multicolumn{3}{l}{$\frac{2}{\sqrt{7}}\begin{pmatrix}\cos \left(\frac{\pi }{14}\right) & \cos \left(\frac{3 \pi }{14}\right) & \sin \left(\frac{\pi }{7}\right) \\
 \cos \left(\frac{3 \pi }{14}\right) & -\sin \left(\frac{\pi }{7}\right) & -\cos \left(\frac{\pi }{14}\right) \\
 \sin \left(\frac{\pi }{7}\right) & -\cos \left(\frac{\pi }{14}\right) & \cos \left(\frac{3 \pi }{14}\right) \end{pmatrix}$
\hypertarget{LY2-SI}{($a$)}}\\
     
   & \multicolumn{2}{l}{~}& \multicolumn{3}{l}{$ \frac{2}{\sqrt{7}}\begin{pmatrix}\cos \left(\frac{3\pi }{14}\right) & \cos \left(\frac{ \pi }{14}\right) & \sin \left(\frac{\pi }{7}\right) \\
 \cos \left(\frac{ \pi }{14}\right) & -\sin \left(\frac{\pi }{7}\right) & -\cos \left(\frac{3\pi }{14}\right) \\
 \sin \left(\frac{\pi }{7}\right) & -\cos \left(\frac{3\pi }{14}\right) & \cos \left(\frac{ \pi }{14}\right) \end{pmatrix}$ \hypertarget{LY2-SII}{($b$)}}\\
    \cmidrule(lr){2-6} 
\rowcolor{mColor1} 1. & $\frac{236}{7}$ & $\frac{16}{7}$ & $\frac{17}{7}$ & $\begin{pmatrix}
 
 0 & \frac{5}{17} & \frac{1}{17} \\
 715139 & 220 & 46 \\
 848656 & 145 & 32 
  
\end{pmatrix}$ &\hyperlink{LY2-SII}{($b$)} \\ 
 2. &  $\frac{4}{7}$  &  $\frac{1}{7}$ & $\frac{3}{7} $ & $\begin{pmatrix}
 
 1 & 45954 & 2925 \\
 1 & -74 & -55 \\
 1 & -1702 & 325 \\
  
\end{pmatrix}$ &  \hyperlink{LY2-SI}{($a$)}\\
\rowcolor{mColor1} 3. & $\frac{164}{7}$ & $\frac{13}{7}$ & $\frac{11}{7}$ & $\begin{pmatrix}
 
 41& 5  & 17 \\
 50922 & -10 & -782  \\
 4797 & -11& 221  \\

\end{pmatrix}$ &\hyperlink{LY2-SI}{($a$)}\\
 4. & $\frac{12}{7}$  &  $\frac{2}{7}$ & $\frac{3}{7} $ & $\begin{pmatrix}
 
 6 & 22990 & 3510 \\
 3 & -132 & -117 \\
 2 & -627 & 378 \\
  
\end{pmatrix}$ &\hyperlink{LY2-SII}{$(b)$}\\
\rowcolor{mColor1} 5. & $\frac{156}{7}$ & $\frac{12}{7}$ & $\frac{11}{7} $ & $\begin{pmatrix}
 
 78  & 9& 10 \\
 27170 & -36 & -285 \\
 5070 & -39  & 210 \\ 
  
\end{pmatrix}$ &\hyperlink{LY2-SII}{$(b)$}\\

6. & $\frac{44}{7}$  &  $\frac{5}{7}$ & $\frac{4}{7} $ &$\left(
\begin{array}{ccc}
 88 & 1564 & 1972 \\
 44 & -184 & -725 \\
 11 & -138 & 348 \\
\end{array}
\right)$ & \hyperlink{LY2-SII}{$(b)$}\\
\rowcolor{mColor1} 7. & $\frac{124}{7}$ & $\frac{9}{7}$ & $\frac{10}{7} $ & $\begin{pmatrix}
 
 248 & 76 & 13 \\
 2108 & -152 & -78 \\
 2108 & -475 & 156 \\
  
\end{pmatrix}$ &\hyperlink{LY2-SII}{$(b)$}\\
8.  & $\frac{52}{7}$  &  $\frac{6}{7}$ & $\frac{4}{7} $ & $\left(
\begin{array}{ccc}
 156 & 475 & 2108 \\
 78 & -152 & -2108 \\
 13 & -76 & 248 \\
\end{array}
\right)$ &\hyperlink{LY2-SI}{($a$)}\\
\rowcolor{mColor1} 9. & $\frac{116}{7}$ & $\frac{8}{7}$ & $\frac{10}{7}$ & $\begin{pmatrix}
 
 348 & 138 & 11 \\
 725 & -184 & -44 \\
 1972 & -1564 & 88 \\
  
\end{pmatrix}$ &\hyperlink{LY2-SI}{($a$)}\\
  10.  & $\frac{60}{7}$  &  $\frac{8}{7}$ & $\frac{3}{7} $ & $\begin{pmatrix}
 
 210 & 39 & 5070 \\
 285 & -36 & -27170 \\
 10 & -9 & 78 \\
  
\end{pmatrix}$ & \hyperlink{LY2-SI}{($a$)}\\ 
\rowcolor{mColor1} 11. & $\frac{108}{7}$ & $\frac{6}{7}$ & $\frac{11}{7}$ & $\begin{pmatrix}
 
 378 & 627 & 2 \\
 117 & -132 & -3 \\
 3510 & -22990 & 6 \\
  
\end{pmatrix}$ &\hyperlink{LY2-SI}{($a$)}\\
 12. & $\frac{68}{7}$  &  $\frac{9}{7}$ & $\frac{3}{7} $ & $\left(
\begin{array}{ccc}
 221 & 11 & 4797 \\
 782 & -10 & -50922 \\
 17 & -5 & 41 \\
\end{array}
\right)$ & \hyperlink{LY2-SII}{$(b)$}\\ 
\rowcolor{mColor1} 13. & $\frac{100}{7}$ & $\frac{5}{7}$ & $\frac{11}{7}$ & $ \begin{pmatrix} 
 325 & 1702 & 1 \\
 55 & -74 & -1 \\
 2925 & -45954 & 1 
\end{pmatrix}$ &\hyperlink{LY2-SII}{$(b)$}\\
14. & $\frac{100}{7}$ & $\frac{12}{7}$ & $\frac{4}{7}$ & $ \begin{pmatrix}

 380 & -1 & 1247 \\
 11495 & 1 & -260623 \\
 55 & -1 & -129 \\
  
\end{pmatrix}$ &\hyperlink{LY2-SII}{$(b)$}\\
\bottomrule
\caption{$(3,0)$ models of type $LY_2$. There are two kinds of $S$-matrices.}
\label{tab(3,0)LY2-2}
\end{longtable}

\begin{longtable}{c||ccc|c|c} \toprule
S. No.  & $c$& $h_1$ &$h_2$&$\mathcal{Y}$ matrix &S matrix\\  
\midrule
 
\rowcolor{mColor2} VI.& \multicolumn{4}{c}{Type $(LY_1)^{\otimes 2}$} & \\ 
& \multicolumn{4}{c}{Multiplicities: $m_1=1,\ m_2=$2} &\\
    & \multicolumn{2}{l}{S matrix:}& \multicolumn{3}{l}{$\frac{1}{\sqrt{5}} \begin{pmatrix}
 \frac{1}{2} \left(\sqrt{5}+1\right) & \frac{1}{2} \left(\sqrt{5}-1\right) & 2 \\
 \frac{1}{2} \left(\sqrt{5}-1\right) & \frac{1}{2} \left(\sqrt{5}+1\right) & -2 \\
  1 & -1 & -1 
\end{pmatrix}$\hypertarget{LY1s-SI}{(a)}}\\
& \multicolumn{2}{l}{~}& \multicolumn{3}{l}{$\frac{1}{\sqrt{5}} \begin{pmatrix}
 \frac{1}{2} \left(\sqrt{5}-1\right) & \frac{1}{2} \left(\sqrt{5}+1\right) & 2 \\
 \frac{1}{2} \left(\sqrt{5}+1\right) & \frac{1}{2} \left(\sqrt{5}-1\right) & -2 \\
  1 & -1 & 1 
\end{pmatrix}$\hypertarget{LY1s-SII}{(b)}} \\
   \cmidrule(lr){2-6} 

\rowcolor{mColor1}  & $\frac{164}{5}$ & $ \frac{12}{5}$ &  $\frac{11}{5}$ & $\left(
\begin{array}{ccc}
 0 & \frac{1}{11} & \frac{10}{11} \\
 615164 & 28 & 220 \\
 254200 & 25 & 224
\end{array}
\right)$ &\hyperlink{LY1s-SI}{(a)}\\
 & $\frac{4}{5}$  &  $\frac25$ &  $\frac15 $ & $\begin{pmatrix}

 2 & 3249 & 47500 \\
 1 & 380 & -1250 \\
 1 & -57 & -130 

\end{pmatrix}$ &\hyperlink{LY1s-SI}{(a)}\\
\rowcolor{mColor1}& $\frac{116}{5}$ & $\frac85$ & $\frac95 $ & $\begin{pmatrix}

 58 & 11 & 10 \\
 4959 & 220 & -30 \\
 27550 & -275 & -26

\end{pmatrix}$ &\hyperlink{LY1s-SI}{(a)}\\
 & $\frac{12}{5}$  &  $\frac15 $ &  $\frac35 $ & $\begin{pmatrix}

 3 & 42483 & 2550 \\
 3 & 27 & -50 \\
 5 & -2295 & 222 

\end{pmatrix}$ &\hyperlink{LY1s-SII}{(b)}\\
\rowcolor{mColor1} & $\frac{108}{5}$  &  $\frac95 $ &  $\frac75 $ & $\left(
\begin{array}{ccc}
 27 & 3 & 50 \\
 42483 & 3 & -2550 \\
 2295 & -5 & 222 
\end{array}
\right)$ &\hyperlink{LY1s-SII}{(b)}\\
  & $\frac{28}{5}$  &  $\frac45$ &  $\frac25 $ & $\begin{pmatrix}

 28 & 676 & 16250 \\
 49 & 104 & -5000 \\
 7 & -26 & 120 

\end{pmatrix}$ &\hyperlink{LY1s-SII}{(b)}\\
\rowcolor{mColor1} & $\frac{92}{5}$ & $\frac65$&$\frac85 $ & $\begin{pmatrix}

 92 & 119 & 10 \\
 1196 & 136 & -20 \\
 7475 & -1700 & 24

\end{pmatrix}$ &\hyperlink{LY1s-SII}{(b)}\\
  & $\frac{36}{5}$  &  $\frac35$ &  $\frac45 $ & $\begin{pmatrix}

 144 & 1452 & 1100 \\
 12 & 336 & -150 \\
 45 & -770 & -228 

\end{pmatrix}$ &\hyperlink{LY1s-SI}{(a)}\\
\rowcolor{mColor1} & $\frac{84}{5}$ & $\frac75$ &  $\frac65 $ & $\left(
\begin{array}{ccc}
 336 & 12 & 150 \\
 1452 & 144 & -1100 \\
 770 & -45 & -228 
\end{array}
\right)$ &\hyperlink{LY1s-SI}{(a)}
\\

& $\frac{44}{5}$  &  $\frac{2}{5}$ & $\frac65 $ & $\begin{pmatrix}

 220 & 4959 & 30 \\
 11 & 58 & -10 \\
 275 & -27550 & -26 \\

\end{pmatrix}$ &\hyperlink{LY1s-SI}{(a)}\\
\rowcolor{mColor1} & $\frac{76}{5}$ & $\frac85$ &  $\frac45 $ & $\begin{pmatrix}

 380 & 1 & 1250  \\
3249 & 2  & -47500 \\
 57  & -1 & -130 \\

\end{pmatrix}$ &\hyperlink{LY1s-SI}{(a)}\\
 & $\frac{52}{5}$  &  $\frac65$ &  $\frac35 $ & $\begin{pmatrix}

 104 & 49 & 5000 \\
 676 & 28 & -16250 \\
 26 & -7 & 120 \\

\end{pmatrix}$ &\hyperlink{LY1s-SII}{(b)}\\
\rowcolor{mColor1} & $\frac{68}{5}$ & $\frac45$ & $\frac75 $ & $\begin{pmatrix}

 136 & 1196 & 20 \\
 119 & 92 & -10 \\
 1700 & -7475 & 24 \\

\end{pmatrix}$ &\hyperlink{LY1s-SII}{(b)}\\
\bottomrule
\caption{$(3,0)$ models of type $(LY_1)^{\otimes 2}$.}
\end{longtable}

\subsubsection{Some observations}

It turns out that in some cases, one or both of the $\mathbf{Y}_i$ that we construct starting from a known RCFT vvmf $\mathbb{X}$ turn out to be the RCFT after some rearrangement and change of signs. We list a few of them.

Let $i=1$ or $2$ and $j\neq i,0$. Consider the exponents of the quasi-character, $(\alpha_0+1,\ \alpha_i-1, \alpha_j)$. Suppose that $\alpha_i-1$  is the lowest exponent and define
\begin{equation*}
    \hat{\mathbb{Y}}^{(i)}: \begin{pmatrix}
    q^{\alpha_i+1}(1+y_{ii}q+\cdots)\\\frac{|y_{1i}|} {y_{1i}}\ q^{\alpha_0+1}(y_{1i}+\mathcal{O}(q))\\ \frac{|y_{ji}|} {y_{ji}}q^{\alpha_j}(y_{ji}+\mathcal{O}(q)))
\end{pmatrix}
\end{equation*}
has $\hat{c}^{(i)}=c+24(1-h_i),\ \hat{h}^{(i)}_i=2-h_i,\ \hat{h}^{(i)}_j=1-h_i+h_j$. 
In the  cases given below, $\hat{\mathbb{Y}}^{(i)}$ is associated with another RCFT with the above $\hat{c}$ and $\hat{h}$:
\begin{enumerate}
\item In the $B_{r,1}$ RCFTs for $r\in[0,11]$, $\hat{\mathbb{Y}}^{(1)}= (\chi_1^{(1)},\chi_0^{(1)},-\chi_2^{(1)})^T$ has $\hat{c}^{(1)}=\frac{2r+25}{2},\ \hat{h}_1=\frac{3}{2},\ \hat{h}^{(1)}_2=\frac{2r+9}{16}\implies \widetilde{B_{11-r,1}}$ RCFT.
\item In the  $D_{r,1}$ RCFT's for $r\in[1,11]$, $\hat{\mathbb{Y}}^{(1)}= (\chi_1^{(1)},\chi_0^{(1)},-\chi_2^{(1)})^T$ has $\hat{c}^{(1)}=r+12,\ \hat{h}_1=\frac{3}{2},\ \hat{h}^{(1)}_2=\frac{r+4}{8}\implies \widetilde{D_{12-r,1}}$ RCFT.
\item One has $\hat{\mathbb{Y}}^{(2)}= (\chi_2^{(2)},-\chi_1^{(2)},\chi_0^{(2)})^T$ in $LY_2$ class and $\hat{\mathbb{Y}}^{(1)}= (\chi_1^{(1)},\chi_0^{(1)},-\chi_2^{(1)})^T$  in $LY_1^{\otimes2}$ class, maps to another example in the same class with appropriate values of $c$ and $h$
 \item In $LY_1^{\otimes 2},A_{2,1}^{\otimes2}, \text{ and }A_{4,1}$, if the diagonal entries of the $S$-matrix $S_{ii}>0$ and $m_i\neq1$, it agrees with theories with appropriate $\hat{c}$ and $\hat{h}$ up to some rescaling. \hl{If $S_{ii}<0$, the $\hat{Y}_i$ cannot be associated with an RCFT as one has $S_{00}>0$.}
\end{enumerate}

\subsection{Admissible characters}

In the previous sub-section, we obtained explicit formulae for quasi-characters starting from all known $(3,0)$ CFT listed in Das-Gowdigere-Mukhi\cite{Das:2022uoe}. The next step is to look for admissible characters that can be obtained using linear combinations of the quasi-characters along with the original $(3,0)$ theory. Our main goal is to list out all $(3,6)$ admissible characters obtained in this fashion. We will also obtain theories with $(3,12)$ and $(3,18)$ along the way. The explicit solutions are listed in the Tables given in Appendix \ref{36Admissible}.

\subsubsection{Admissible characters of type $U$}

It is easy to track the change in the Wronskian index of admissible characters. One has $\delta \ell = -6 \sum_a\delta \alpha_a$ -- in other words, the change in Wronskian idex is determined by the studying the leading exponents of the admissible character. Starting with
an RCFT with exponents $(\alpha_1,\alpha_2,\alpha_3)$, one has
\begin{equation}
\mathbb{U}_2(b_1,b_2) = \begin{pmatrix}
q^{\alpha_0} (1 + O(q)) \\
q^{\alpha_1-1} (b_1 + p_{1,1}(b_1,b_2)\, q + p_{1,2}(b_1,b_2)\,q^2+O(q^3))\\
q^{\alpha_2-1} (b_2 + p_{2,1}(b_1,b_2)\, q + p_{2,2}(b_1,b_2)\,q^2+O(q^3))
\end{pmatrix}\quad,
\end{equation}
where $p_{i,j}$ are linear polynomials in $(b_1,b_2)$. For admissible solutions, one imposes the conditions
$b_i\geq0$ and $p_{i,j}\geq0$. Since $ \sum_a\delta \alpha_a=-2$, we see that $\delta\ell=12$. Thus, starting from a $(3,0)$ theory, an admissible character obtained as $\mathbb{U}_2$ will have Wronskian index $12$.

Next consider the case of $\mathbb{U}_{1,1}$, where $b_2=0$ with $b_1>0$, then one has $\sum_a\delta \alpha_a = -1$  leading to an admissible character with Wronskian index $6$ when
$p_{2,1}(b_1,0)>0$. However, if $p_{2,1}(b_1,0)=0$ for an integral value of $b_1$, then we obtain an admissible character with vanishing Wronskian index.

The set of admissible characters of type $U$ are integral points that lie in the interior of a bounded convex polytope. Interior points are  $(3,12)$ characters, while those on the boundaries $b_1=0$ and $b_2=0$ will be typically $(3,6)$ while those at the corners of the polytope (if they are integral points) will be $(3,0)$. This is illustrated with an example in Fig. \ref{UExample}.

\begin{figure}[ht]
\centering
\includegraphics[scale=0.7]{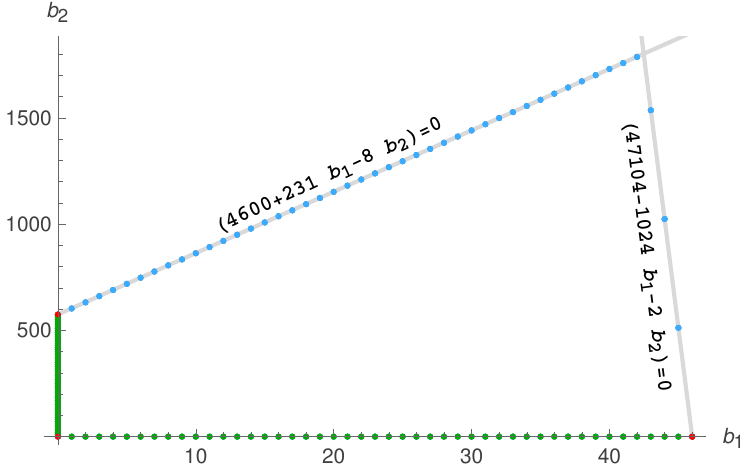}
\caption{Example of type $U$ admissible characters starting with a $c = 23$, $h_1= 3/2$, $h_2=15/8$ RCFT. They occur as integral points on the interior and boundary of a quadrilateral. We see two other $(3,0)$ theories at the corners $(b_1,b_2)=(46,0)$ and
$(b_1,b_2)=(0,575)$. These are shown as red dots and connected by green dots which are $(3,6)$ admissible characters. We show the $(3,12)$ points that lie on the boundaries as blue dots. All integral interior points are $(3,12)$ admissible characters and are not shown in the figure to avoid clutter.}\label{UExample}
\end{figure}

\subsubsection{Admissible characters of type $W$}

\begin{equation}
\mathbb{W}_2(b,c_1,c_2) = \begin{pmatrix}
q^{\alpha_0-1} (1 + b\, q +O(q^2)) \\
q^{\alpha_1-1} (c_1 + p_{1,1}(b,c_1,c_2)\, q + p_{1,2}(b_1,b_2)\,q^2+O(q^3))\\
q^{\alpha_2-1} (c_2 + p_{2,1}(b,c_1,c_2)\, q + p_{2,2}(b,c_1,c_2)\,q^2+O(q^3))
\end{pmatrix}\quad,
\end{equation}
We list the various possibilities that depend on specific choices for the free parameters $(b,c_1,c_2)$. 
\begin{enumerate}
\item
Since $ \sum_a\delta \alpha_a=-3$, we have $\delta \ell=18$ when $c_1,c_2\neq 0$. 
These are of type $\mathbb{W}_2$.
\item However, when one of $c_1$ or $c_2$ vanishes, one has $\delta \ell=12$ generically. However, when $c_2=0$ and $p_{2,1}(b,c_1,c_2=0)=0$, we obtain $\delta \ell=6$. These are of type $\mathbb{W}_{1,1}$	.
\item  If additionally, we can set $p_{2,2}(b,c_1,c_2=0)=0$, we can obtain theories with $\delta \ell=0$. This happens only in a limited set of examples.
\item If $c_1=c_2=0$, we obtain admissible solutions for $b\geq b_{\text{min}}$. These are of type $\mathbb{W}_0$ are generically with $\delta \ell=6$.
\end{enumerate}
An infinite number of admissible theories appear inside a non-compact polytope in the space $(b,c_1,c_2)$ with faces given by $c_1=0$, $c_2=0$, $b=b_{\text{min}}$ and $p_{i,m}=0$. This is unlike those of type $U$ where we obtained only a finite number of admissible solutions as the polytope was compact.

\begin{figure}[ht]
\centering
\includegraphics[scale=1]{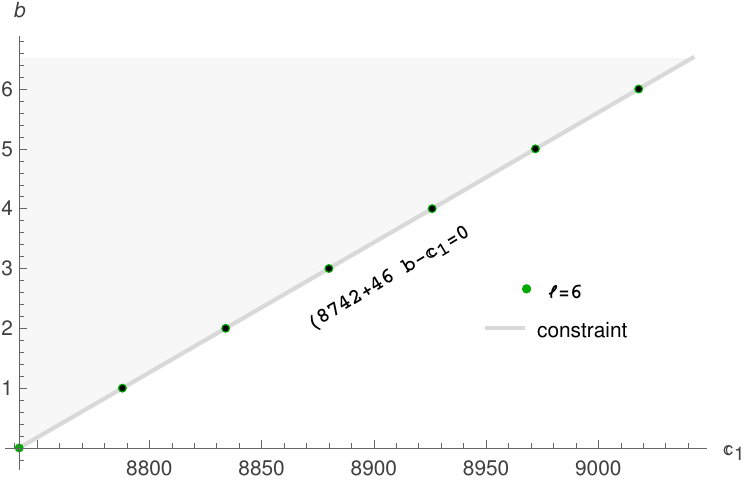}
\caption{In this type $\mathbb{W}_{1,1}$ example, we obtain $(3,6)$ theories on integral points that lie on the line $\frac1{46}(c_1-8742)$. All integral points inside the polytope and on the $c_1=0$ line are $(3,12)$ admissible solutions. There are no $(3,0)$ solutions.}
\end{figure}

\subsubsection{Admissible characters of type $V$}

\begin{equation}
\mathbb{V}_2 = \begin{pmatrix}
q^{\alpha_0-2} (1 + b_1\, q + b_2 q^2 +O(q^3)) \\
q^{\alpha_1-1} (s_1 + p_{1,1}(b,s_1,s_2)\, q + p_{1,2}(b_1,b_2)\,q^2+O(q^3))\\
q^{\alpha_2-1} (s_2 + p_{2,1}(b,s_1,s_2)\, q + p_{2,2}(b,s_1,s_2)\,q^2+O(q^3))
\end{pmatrix}\quad,
\end{equation}
\begin{enumerate}
\item When $s_1,s_2>0$, we have $\delta \ell=24$. This is of type $\mathbb{V}_2$.
\item When one of $(s_1,s_2)$ is zero with the other being non-zero, we have $\delta \ell=18$. When $s_1>0$, we look for additional conditions on the constants $(s_2,b_1,b_2)$ such that $p_{2,1}=0$ which will lead to $\delta \ell =12$. One can get $\delta \ell =6,0$ at higher co-dimension.
\item When $s_1=s_2=0$, we obtain $\delta \ell =12$. However, by imposing vanishing conditions, at codimension two, we can obtain $\delta \ell=0$.
\end{enumerate}

\hl{
\noindent\textbf{Remarks:} 
\begin{enumerate}
\item Solutions of type $U$ have two positive parameters. For sufficiently large values of these parameters, the solution will \textit{not} be admissible. So admissible solutions lie on a polytope in a bounded region. We anticipate but have not proven that the polytope is convex. 
\item  Solutions of type $W$ (resp. $V$) has three (resp. four) positive parameters that are not naturally bounded. We anticipate but have not proven that admissible solutions lie on a unbounded convex polytope. 
\item The Wronskian index of the admissible solutions depends on the dimension of the \textit{face} of the polytope that it belongs to. Those that lie in the interior have the highest Wronskian index. The Wronksian index of admissible solutions lying on co-dimension $d$ face, will be smaller by $6d$. This follows trivially by looking at the leading exponents of the solution. Thus for admissible solutions of type $U$, we can have $\ell=12, 6,0$. For admissible solutions of type $W$,  we can have $\ell=18,12, 6,0$ and for type $V$,  we can have $\ell=,24,18,12, 6,0$.
\end{enumerate}
}

\section{$(3,3)$ theories}

We consider the $(3,3)$ admissible characters obtained using the MLDE method. For a $(3,3)$ CFT, one has $\sum_{i=0}^2\alpha_i=0$. Further, one has
\[
3a_1+ 2(b_1+2b_2)=12\ .
\]
$a_1=2$ and $b_1=b_2=1$ is the only indecomposable solution. For instance,  $a_1=0$ is not allowed as the $S$-matrix is then diagonal and hence the theory is a tensor product of three single-character theories. Thus one has $\det(S)=-1$ and $\det(U)=1$. Further, $\det(T)=1$ which is consistent with $\sum_i\alpha_i=0$. 

\subsection{A new construction of $(3,3)$ admissible characters}

Let $\Xi(\tau)$ represent the basis of VVMFs with a given multiplier system. Further, let $(\Lambda,\mathcal{Y})$ be the data associated with the matrix MLDE associated with $\Xi(\tau)$. Bantay and Gannon introduce the notion of a dual of $\Xi(\tau)$ that is defined by the involutive transformation\cite{Bantay:2007zz}
\begin{equation}
\Xi^\vee(\tau) := \frac{E_4(\tau)^2 E_6(\tau)}{\Delta(\tau)^{7/6}} \left(\Xi(\tau)^T\right)^{-1}\ .
\label{BFdual}
\end{equation}
The data associated with $\Xi^\vee(\tau)$ is the given by
\begin{equation}
\begin{split}
\Lambda^\vee &= -\frac{7}{6}I_n - \Lambda\ , \\
\mathcal{Y}^\vee &= 4 I_n -\mathcal{Y}^T \ .
\end{split}
\end{equation}
Let $(n,\ell)$ denote the rank and Wronskian index of $\Xi(\tau)$. Then $\Xi^\vee(\tau)$ has rank $n$ and Wronskian index given by
\begin{equation}
\ell^\vee = (n-5)(n-1)+7 -\ell\ .
\end{equation}
The $S$-matrix for $\Xi^\vee$ is given by $S^\vee=-(S^{-1})^T$, where $S$ is the $S$-matrix of the original theory. The minus sign arises from the one-dimensional character, $\frac{E_4^2 E_6}{\Delta^{7/6}}$. Next, let  $D=\text{Diag}(1,m_1,\ldots,m_{n-1})$ be the matrix of multiplicities of the original theory. Then, $D^\vee  \propto D^{-1}$, the proportionality constant is fixed by requiring the vacuum character to have multiplicity one.

This involution relates $(2,0)$ characters to $(2,4)$ theories and maps $(2,2)$ to itself. What is interesting is that this involution maps $(3,0)$ VVMF's to $(3,3)$ VVMFs, thus providing a new construction of $(3,3)$ admissible characters. Not all $(3,3)$ VVMF's obtained this way turn out to be admissible. However, \textbf{all} known $(3,3)$ theories arise in this fashion. An additional bonus of this construction is that we can determine the $S$-matrices of known $(3,3)$ admissible characters along with multiplicities that cannot be obtained from the MLDE approach. 

We illustrate this with example of the Ising model. A straightahead computation gives the following result (we show only the first few terms).
\[
\Xi^\vee =\begin{pmatrix}
 \frac{1}{q^{55/48}} + \frac{4}{q^{7/48}} & -\frac{1}{q^{7/48}}-3826 q^{41/48} & -\frac{1}{q^{7/48}} -53 q^{41/48} \\
 -2325 q^{17/48}-1785755 q^{65/48} & \frac{1}{q^{31/48}}-271 q^{17/48} & 25 q^{17/48}+246 q^{65/48}\\
 -94208 q^{19/24}-21581824 q^{43/24} & 4096 q^{19/24} +208896 q^{43/24}&\frac{1}{q^{5/24}}+27 q^{19/24}
\end{pmatrix}
\]
The VVMF corresponding to the third column becomes admissible on switching the sign of the first term. This is the $(3,3)$ admissible solution with $c=5$ and $h_1=1/16$ and $h_2=9/16$. One needs to permute the rows to match with solution S1 given in \cite{Gowdigere:2023xnm}. In computing the S-matrix, we need to account for the changes in sign and permutations. In Table \ref{S33}, we give the results for the $15$ $(3,3)$ admissible characters.

\renewcommand{\arraystretch}{1.2}
\begin{longtable}{c|ccc||ccc}
\toprule
&\multicolumn{3}{c||}{$(3,3)$ theory}&\multicolumn{3}{c}{Dual $(3,0)$ Theory}\\
S.No. &$c$& $h_1$&  $h_2$ & $\tilde{c}$& $\tilde{h}_1$ &  $\tilde{h}_2$ \\  
\midrule
\multicolumn{7}{c}{ $m_1=m_2=1$} \\
\multicolumn{7}{c}{$S:$ 
   $ \left(
\begin{array}{ccc}
 0 & \frac{1}{\sqrt{2}} & \frac{1}{\sqrt{2}} \\
 \frac{1}{\sqrt{2}} & -\frac{1}{2} & \frac{1}{2} \\
 \frac{1}{\sqrt{2}} & \frac{1}{2} & -\frac{1}{2} \\
\end{array}
\right)$\ ,}  \\
\midrule
  1& $5$ & $\frac{1}{16}$ & $\frac{9}{16}$ &$\begin{array}{c}
 \frac{1}{2}\\\frac{25}{2}\end{array}$&$\begin{array}{c}
 \frac{1}{2}\\\frac{3}{2}\end{array}$ &$\begin{array}{c}
 \frac{1}{16}\\\frac{9}{16}\end{array}$ \\
\rowcolor{mColor1}2& $7$ &$\frac{3}{16}$ & $\frac{11}{16}$ &$\begin{array}{c}
 \frac{3}{2}\\\frac{27}{2}\end{array}$&$\begin{array}{c}
 \frac{1}{2}\\\frac{3}{2}\end{array}$ &$\begin{array}{c}
 \frac{3}{16}\\\frac{11}{16}\end{array}$\\
 \midrule
 \multicolumn{7}{c}{ $m_1=m_2=2$} \\
 \rowcolor{mColor3} \multicolumn{7}{c}{$S:$  $\left(
\begin{array}{ccc}
 0 & 1 & 1 \\
 \frac{1}{2} & -\frac{1}{2} & \frac{1}{2} \\
 \frac{1}{2} & \frac{1}{2} & -\frac{1}{2} \\
\end{array}
\right)$}\\
\midrule
3& $6$ & $\frac18$ & $\frac58$ &$\begin{array}{c}
1\\13\end{array}$&$\begin{array}{c}
 \frac{1}{2}\\\frac{3}{2}\end{array}$&$\begin{array}{c}
 \frac{1}{8}\\\frac{5}{8}\end{array}$\\
\rowcolor{mColor1}4&  $8$ &$\frac{1}{4}$ & $\frac{3}{4}$ &$\begin{array}{c}
2\\14\end{array}$&$\begin{array}{c}
 \frac{1}{2}\\\frac{3}{2}\end{array}$&$\begin{array}{c}
 \frac{1}{4}\\\frac{3}{4}\end{array}$\\
 \midrule

\rowcolor{mColor3} \multicolumn{7}{c}{ $m_1=m_2=2$} \\
 \multicolumn{7}{c}{$S:$ $\left(
\begin{array}{ccc}
 \frac{1}{\sqrt{5}} & \frac{2}{\sqrt{5}} & \frac{2}{\sqrt{5}} \\
 \frac{1}{\sqrt{5}} & -\frac{1}{10} \left(\sqrt{5}+5\right) & \frac{1}{10} \left(5-\sqrt{5}\right) \\
 \frac{1}{\sqrt{5}} & \frac{1}{10} \left(5-\sqrt{5}\right) &- \frac{1}{10} \left(\sqrt{5}+5\right) \\
\end{array}
\right)$}\\
\midrule
  5& $8$ &$\frac15$ &$\frac45$ &$\begin{array}{c}
 \frac{4}{5}\\\frac{84}{5}\end{array}$&$\begin{array}{c}
 \frac{1}{5}\\\frac{6}{5}\end{array}$ &$\begin{array}{c}
 \frac{2}{5}\\\frac{7}{5}\end{array}$ \\ 
\rowcolor{mColor1}6& 16 &$\frac45$ &$\frac65$ &$\begin{array}{c}
 \frac{36}{5}\\\frac{76}{5}\end{array}$&$\begin{array}{c}
 \frac{4}{5}\\\frac{4}{5}\end{array}$ &$\begin{array}{c}
 \frac{3}{5}\\\frac{8}{5}\end{array}$\\
  7& 24&$\frac15$ &$\frac45$ &$\begin{array}{c}
 \frac{44}{5}\\\frac{116}{5}\end{array}$&$\begin{array}{c}
 \frac{6}{5}\\\frac{9}{5}\end{array}$ &$\begin{array}{c}
 \frac{2}{5}\\\frac{8}{5}\end{array}$ \\ 
\rowcolor{mColor1}8& 32 &$\frac45$ &$\frac65$&$\frac{76}{5}$&$\frac{9}{5}$ &$ \frac{3}{5}$ \\
\midrule
\rowcolor{mColor3}  \multicolumn{7}{c}{ $m_1=m_2=1$} \\
 \multicolumn{7}{c}{$S:$ $\frac{2}{\sqrt{7}}\left(
\begin{array}{ccc}
 \sin \left(\frac{\pi }{7}\right) & \cos \left(\frac{\pi }{14}\right) & \cos \left(\frac{3 \pi }{14}\right) \\
 \cos \left(\frac{\pi }{14}\right) & -\cos \left(\frac{3 \pi }{14}\right) & \sin \left(\frac{\pi }{7}\right) \\
 \cos \left(\frac{3 \pi }{14}\right) & \sin \left(\frac{\pi }{7}\right) & -\cos \left(\frac{\pi }{14}\right) \\
\end{array}
\right)$}\\
\midrule 
9 & $\frac{48}{7}$ &$\frac{5}{7}$ &$\frac{1}{7}$ &$\begin{array}{c}
 \frac{4}{7}\\\frac{100}{7}\end{array}$&$\begin{array}{c}
 \frac{1}{7}\\\frac{5}{7}\end{array}$ &$\begin{array}{c}
 \frac{3}{7}\\\frac{11}{7}\end{array}$ \\
\rowcolor{mColor1}10& $\frac{64}{7}$ &$\frac{2}{7}$ &$\frac{6}{7}$ &$\begin{array}{c} \frac{12}{7} \\ \frac{108}{7}\end{array}$&$\begin{array}{c}
 \frac{2}{7}\\\frac{6}{7}\end{array}$ &$\begin{array}{c}
 \frac{3}{7}\\\frac{11}{7}\end{array}$ \\
11& $\frac{104}{7}$ &$\frac{5}{7}$ &$\frac{8}{7}$ &$\begin{array}{c}
 \frac{44}{7}\\\frac{116}{7}\end{array}$&$\begin{array}{c}
 \frac{4}{7}\\\frac{8}{7}\end{array}$ &$\begin{array}{c}
 \frac{5}{7}\\\frac{10}{7}\end{array}$\\
\rowcolor{mColor1}12& $\frac{120}{7}$ &$\frac{9}{7}$ &$\frac{6}{7}$ &$\begin{array}{c}
 \frac{52}{7}\\\frac{124}{7}\end{array}$&$\begin{array}{c}
 \frac{4}{7}\\\frac{9}{7}\end{array}$ &$\begin{array}{c}
 \frac{6}{7}\\\frac{10}{7}\end{array}$ \\
\bottomrule	
13& $\frac{160}{7}$ &$\frac{12}{7}$ &$\frac{8}{7}$ &$\begin{array}{c}
 \frac{60}{7}\\\frac{156}{7}\end{array}$&$\begin{array}{c}
 \frac{3}{7}\\\frac{11}{7}\end{array}$ &$\begin{array}{c}
 \frac{8}{7}\\\frac{12}{7}\end{array}$ \\
\rowcolor{mColor1}14& $\frac{176}{7}$ &$\frac{9}{7}$ &$\frac{13}{7}$ &$\begin{array}{c}
 \frac{68}{7}\\\frac{164}{7}\end{array}$&$\begin{array}{c}
 \frac{3}{7}\\\frac{11}{7}\end{array}$ &$\begin{array}{c}
 \frac{9}{7}\\\frac{13}{7}\end{array}$\\
15& $\frac{216}{7}$ &$\frac{12}{7}$ &$\frac{15}{7}$ &$\frac{100}{7}$&$
 \frac{4}{7}$ &$\frac{12}{7}$\\

\bottomrule	
\caption{We provide the S-matrix and multiplicities for all known $(3,3)$ admissible characters. The third column gives the dual $(3,0)$ theories -- there can be more than one when $\Xi$ leads to multiple admissible characters. }
\label{S33}
\end{longtable}

Now that we know the S-matrix of the $(3,3)$ admissible characters, we can see if they can represent RCFTs. 
The next step is to determine the fusion rules and see if they lead to positive, integral coefficients. For the cases when the multipicities are one, it is somewhat straightforward. Even in these cases, we have consider two possibilities -- the theory is unitary or non-unitary. To account for non-unitary theories, we allow for the identity operator to be the operator with smallest dimension. We find the following:
\begin{enumerate}
\item For items 1 and 2 in Table \ref{S33}, the theories cannot be unitary as $S_{00}=0$. Allowing for non-unitary theories leads to fusion coefficients being negative. Thus, these two cannot be RCFTs.
\item For items 9-15, we find acceptable fusion rules compatible with a unitary RCFT. The fusion rules are given in Table \ref{fusionrule}.
\begin{table}[h]
\centering
\begin{tabular}{c||c|c|c}
& $1$ & $\phi_1$ & $\phi_2$ \\  \hline\hline
$1$ & $1$ & $\phi_1$ & $\phi_2$ \\
$\phi_1$ & $\phi_1$ & $1+\phi_1+\phi_2$ & $\phi_1+\phi_2$ \\ 
$\phi_2$ & $\phi_2$ & $\phi_1+\phi_2$ & $1+\phi_1$ \\ \hline
\end{tabular}
\caption{The three primary operators are labelled $(1,\phi_1,\phi_2)$ and the fusion rules are given as a multiplication table.}\label{fusionrule}
\end{table}
\item
We next consider the remaining theories. The two familes of $S$-matrices needs to be first resolved into $5\times 5$ matrices. There are three undetermined coefficients which are determined by the condition that $S_{\text{resolved}}\cdot S_{\text{resolved}}=C$, where $C$ satisfies $C^2=1$ with $C$ being the identity matrix being the simplest possibility. It turns out that \textit{none} of the solutions lead to \hl{non-negative fusion} coefficients. Thus, we believe that these are not associated with RCFTs.
\end{enumerate}

\noindent \textbf{Remarks:} \\
1. Using a different method, Gowdigere-Kala-Rawat have obtained the $S$-matrix and fusion coefficients for the $15$ admissible characters\cite{GKR}. Their results are consistent with the one discussed here.\\
2. Consider the product theories $E_{7,1}\otimes A_{1,1}$ and $E_{7\frac{1}{2}}\otimes LY_{1}$ with $c=8$ and $h_i=(0,\frac{1}{4},\frac{3}{4},1),\ (0,\frac{1}{5},\frac{4}{5},1)$. The three-character theory can be obtained from these four-character RCFTs as follows:
\begin{equation*}
    \mathbb{X}(\tau)=(\chi_0-\chi_3,\chi_1,\chi_2)^T\ .
\end{equation*} 
They correspond to $(3,3)$ theories with $c=8$ and $h_i=(0,\frac{1}{4},\frac{3}{4}),\ (0,\frac{1}{5},\frac{4}{5})$ respectively. The minus sign in the linear combination $\chi_0-\chi_3$ shows that it cannot be an RCFT\footnote{We thank Kaiwen Sun for pointing out these examples to us.} .

\subsection{$(3,3)$ Quasi-characters and Admissible characters from VVMFs}

Column 3 of  Table \ref{tab_33} provides the $q$-series of the $(3,3)$ admissible characters. The important addition we provide is the degeneracies that we obtain through the duality. For instance, in item 2, the third character has degeneracy $54$ \hl{instead of 27 obtained from integrality constraints}. In column 4 of Table \ref{tab_33}, we list the $\mathcal{Y}$ matrices corresponding to the 15 $(3,3)$ solutions obtained using the MLDE method in \cite{Gowdigere:2023xnm}. The quasi-characters $\mathbb{Y}_1$ and $\mathbb{Y}_2$ can be obtained by solving the matrix modular differential equation of Gannon\cite{Gannon:2013jua}. Once these are obtained, we look for admissible characters of types $U$, $W$ and $V$. The last column of Table \ref{tab_33}, lists the range of values $0\leq b\leq b_{\text{min}}$ for which we obtain $(3,9)$ admissible characters of type $W_0$. Notice that there are cases where we obtain \textit{no} solutions. Table \ref{tabU_33} lists $(3,9)$ admissible characters of type $U_i$, Table \ref{tabW_33} lists $(3,9)$ admissible characters of type $W_{1,i}$ and Table \ref{tabV_33} lists $(3,9)$ admissible characters of type $V_0$ that arise from the 7 solutions that are potentially RCFTs.

\clearpage

\begin{longtable}{c||ccc|c|ll}  

\toprule
&&&&&\\
\multirow{-2}{11pt}{Sl. No.}  & $c$& $h_1$ &  $h_2$& $\mathbb{X}$ & $\mathcal{Y}$ matrix & bmin$_{2}$\\  
\midrule
\rowcolor{mColor1}  1& $5$ & $\frac{1}{16}$ & $\frac{9}{16}$ &$\begin{pmatrix}q^{-\frac{5}{24}}(1 + 27 q + 106 q^2+\dots)\\q^{-\frac{7}{48}}(1 + 53 q + 404 q^2+\dots)\\q^{\frac{17}{48}}(25 + 246 q + 1297 q^2+\dots )\end{pmatrix}$&$\left(
\begin{array}{ccc}
 27 & 94208 & 4096 \\
 1 & 4 & 1 \\
 25 & 2325 & -271 \\
\end{array}
\right)$ &-- \\
\rowcolor{mColor1}2& $7$ &$\frac{3}{16}$ & $\frac{11}{16}$&$\begin{pmatrix}q^{-\frac{7}{24}}(1 + 25 q + 53 q^2+\dots)\\2q^{-\frac{5}{48}}(1 + 106 q + 1214 q^2+\dots)\\2q^{\frac{19}{48}} (27 + 433 q + 3400 q^2+\dots )\end{pmatrix}$& $\left(
\begin{array}{ccc}
 25 & 43008 & 2048 \\
 2 & 1 & 3 \\
 54 & 2871 & -266 \\
\end{array}
\right)$ &-- \\
3& $6$ & $\frac18$ & $\frac58$ &$\begin{pmatrix}q^{-\frac{1}{4}}(1 + 26 q + 79 q^2+\dots)\\q^{-\frac{1}{8}}(1 + 79 q + 755 q^2+\dots)\\2q^{\frac{3}{8}} (13 + 163 q + 1053 q^2+\dots )\end{pmatrix}$ &$\left(
\begin{array}{ccc}
 26 & 90112 & 4096 \\
 1 & 3 & 2 \\
 26 & 2600 & -269 \\
\end{array}
\right)$ & --\\

4& $8$ &$\frac14$ &$\frac34 $& $\begin{pmatrix}q^{-\frac{1}{3}}(1 + 24 q + 28 q^2+\dots)\\2q^{-\frac{1}{12}}(1 + 134 q + 1809 q^2+\dots)\\8 q^{\frac{5}{12}} (7 + 142 q + 1329 q^2+\dots )\end{pmatrix}$ & $\left(
\begin{array}{ccc}
 24 & 40960 & 2048 \\
 2 & -2 & 4 \\
 56 & 3136 & -262 \\
\end{array}
\right)$& --\\

\rowcolor{mColor1}5 & $8$ &$\frac15$ &$\frac45$& $\begin{pmatrix}q^{-\frac{1}{3}}(1 + 134 q + 1920 q^2+\dots)\\q^{-\frac{2}{15}}(1 + 190 q + 2832 q^2+\dots)\\q^{\frac{7}{15}} (57 + 1159 q + 10526q^2+\dots )\end{pmatrix}$ &$\left(
\begin{array}{ccc}
 134 & 47500 & 1250 \\
 1 & 2 & 1 \\
 57 & 3249 & -376 \\
\end{array}
\right)$&$3 ^*$\\

\rowcolor{mColor1}6& $ 16$ &$\frac45$ & $\frac65$& $\begin{pmatrix}q^{-\frac{2}{3}}(1 + 232 q + 31076 q^2+\dots)\\5 q^{-\frac{2}{15}}(9 + 2792 q + 101678 q^2+\dots)\\5 q^{\frac{8}{15}} (154 + 13264 q + \dots )\end{pmatrix}$&$\left(
\begin{array}{ccc}
 232 & 1100 & 150 \\
 45 & -140 & 12 \\
 770 & 1452 & -332 \\
\end{array}
\right)$&0 \\

\rowcolor{mColor1}7& $ 24$ &$\frac65$ & $\frac95$ & $\begin{pmatrix}q^{-1}(1 + 30 q + 87786 q^2+\dots)\\25q^{\frac{1}{5}}(11 + 10212 q +\dots)\\25q^{\frac{4}{5}}(1102 + 166953 q +\dots )\end{pmatrix}$&$\left(
\begin{array}{ccc}
 30 & 30 & 10 \\
 275 & -216 & 11 \\
 27550 & 4959 & -54 \\
\end{array}
\right)$&0 \\

\rowcolor{mColor1}8& $ 32$ &$\frac95$ & $\frac{11}{5}$& $\begin{pmatrix}q^{-\frac{4}{3}}(1 + 3 q + 62500  q^2+\dots)\\625q^{\frac{7}{15}}(19 + 17100 q+\dots)\\625q^{\frac{13}{15}}(434 + 135997 q +\dots )\end{pmatrix}$  &$\left(
\begin{array}{ccc}
 3 & -1 & 1 \\
 11875 & -433 & 57 \\
 271250 & 682 & 190 \\
\end{array}
\right)$&0 \\

9& $\frac{48}{7}$ & $\frac17$ &$\frac57$ &$\begin{pmatrix}q^{-\frac{2}{7}}(1 + 78 q + 784  q^2+\dots)\\q^{-\frac{1}{7}}(1 + 133 q + 1618 q^2+\dots)\\q^{\frac{3}{7}} (55 + 890 q + 6720 q^2+\dots )\end{pmatrix}$& $\left(
\begin{array}{ccc}
 78 & 45954 & 1702 \\
 1 & 3 & 1 \\
 55 & 2925 & -321 \\
\end{array}
\right)$ &$3^*$ \\

10& $ \frac{64}{7}$ &$\frac27$ &$\frac67$ & $\begin{pmatrix}q^{-\frac{8}{21}}(1 + 136 q + 2417 q^2+\dots)\\q^{-\frac{2}{21}}(3 + 632 q + 10787 q^2+\dots)\\q^{\frac{10}{21}} (117 + 2952 q + 32220q^2+\dots )\end{pmatrix}$ & $\left(
\begin{array}{ccc}
 136 & 22990 & 627 \\
 3 & -2 & 2 \\
 117 & 3510 & -374 \\
\end{array}
\right)$&2\\
11& $ \frac{104}{7}$ &$\frac57$ &  $\frac87$ & $\begin{pmatrix}q^{-\frac{13}{21}}(1 + 188 q + 17260 q^2+\dots)\\q^{\frac{2}{21}}(44 + 13002 q + 424040 q^2\dots)\\q^{\frac{11}{21}}(725 + 52316 q+ \dots )\end{pmatrix}$  & $\left(
\begin{array}{ccc}
 188 & 1564 & 138 \\
 44 & -84 & 11 \\
 725 & 1972 & -344 \\
\end{array}
\right)$ & 0\\

12& $ \frac{120}{7}$ &$\frac67$ &  $\frac97 $ & $\begin{pmatrix}q^{-\frac{5}{7}}(1 + 156 q + 28926 q^2+\dots)\\q^{\frac{1}{7}}(78 + 28692 q + 1194804 q^2+\dots)\\q^{\frac{4}{21}}(2108 + 200787 q + \dots )\end{pmatrix}$ &$\left(
\begin{array}{ccc}
 156 & 475 & 76 \\
 78 & -152 & 13 \\
 2108 & 2108 & -244 \\
\end{array}
\right)$&0\\
13& $ \frac{160}{7}$ & $\frac87$ & $\frac{12}{7}$  & $\begin{pmatrix}q^{-\frac{20}{21}}(1 + 40 q + 60440 q^2+\dots)\\q^{\frac{4}{21}}(285 + 227848 q +\dots)\\q^{\frac{16}{21}}(27170 + 3857360 q +\dots )\end{pmatrix}$ & $\left(
\begin{array}{ccc}
 40 & 39 & 9 \\
 285 & -206 & 10 \\
 27170 & 5070 & -74 \\
\end{array}
\right)$&0 \\

14& $ \frac{176}{7}$ & $\frac97$ & $\frac{13}{7} $  & $\begin{pmatrix}q^{-\frac{22}{21}}(1 + 14 q + 66512 q^2+\dots)\\q^{\frac{5}{21}}(782 + 718267 q +\dots)\\q^{\frac{17}{21}}(50922 + 8656740 q +\dots )\end{pmatrix}$ &  $\left(
\begin{array}{ccc}
 14 & 11 & 5 \\
 782 & -217 & 17 \\
 50922 & 4797 & -37 \\
\end{array}
\right)$ & 1\\
15& $ \frac{216}{7}$ &$\frac{12}{7}$ & $\frac{15}{7} $& $\begin{pmatrix}q^{-\frac{9}{7}}(1 + 3 q + 52254 q^2+\dots)\\q^{\frac{3}{7}}(11495 + 10341870 q +\dots)\\q^{\frac{6}{7}}(260623 + 74348634 q +\dots )\end{pmatrix}$  & $\left(
\begin{array}{ccc}
 3 & -1 & 1 \\
 11495 & -376 & 55 \\
 260623 & 1247 & 133 \\
\end{array}
\right)$&0\\

\bottomrule		
\caption{The list of (3,3) admissible solutions. Column 3 gives the $q$-series along with degeneracies that follow from the duality with $(3,0)$ theories. The last column indicates the minimum value of $b$ for which we obtain $(3,9)$ admissible solutions.  In four examples, we indicate the absence of $(3,9)$ admissible solutions by a dash. For two cases,  we obtain other $(3,3)$ solutions when $b=b_{\text{min}}$ and $(3,9)$ solutions when $b>b_{\text{min}}$ . These are indicated by adding a ${}^*$ to the quoted value of $b_{\text{min}}$. } 
\label{tab_33}
\end{longtable}

\begin{table}[H]  
\centering
\aboverulesep = 0pt
\belowrulesep = 0pt
\begin{tabular}{c|ccc|ccc}\toprule

S. No.  & $c$& $h_1$ &  $h_2$  & $b_i$& $h_1$ &$h_2$  \\  
\midrule
  1& $ \frac{104}{7}$ & $\frac57$ & $\frac87$  & $b_2\leq 2$& $\frac{5}7$&$\frac{1}7$\\
\rowcolor{mColor1}2& $ \frac{120}{7}$ & $\frac67$ & $\frac97 $ & $b_2\leq 8$& $\frac67$ &$\frac27 $\\
3& $ \frac{160}{7}$ & $\frac87$ &$\frac{12}{7}$   & $b_1\leq 1$& $\frac17$&$\frac{12}{7}$   \\
& & & & $b_2\leq 367$& $\frac87$ &$\frac{5}{7}$\\
 \rowcolor{mColor1}4& $ \frac{176}{7}$ & $\frac97$ & $\frac{13}{7} $ & $b_1\leq 3$& $\frac27$&$\frac{13}{7} $ \\
\rowcolor{mColor1}& &&& $b_2\leq 1376$& $\frac97$ &$\frac{6}{7} $\\
 5& $ \frac{216}{7}$ & $\frac{12}{7}$ & $\frac{15}{7} $ & $b_1\leq 3$& $\frac{12}{7}$ &$\frac{15}{7} $ \\
& &&& $b_2>0$& $\frac{12}{7}$ &$\frac{8}{7} $\\
 \bottomrule

\end{tabular}
\caption{$(3,9)$ admissible solutions of type $U_i$  using $(3,3)$ solutions .}
\label{tabU_33}
\end{table}

\begin{table}[H]  
\centering
\aboverulesep = 0pt
\belowrulesep = 0pt
\begin{tabular}{c|ccc|cc|ccc}\toprule

S. No.  & $c$& $h_1$ &  $h_2$  & $c_i$& $b$&  $c$&$h_1$ &$h_2$  \\  
\midrule
1.& $\frac{48}{7}$ & $\frac17$ &$\frac57$    & $c_2=3-b$& $0 \leq b <3^*$&  $\frac{216}{7}$&$\frac{15}{7}$&$\frac{5}{7}$\\
\rowcolor{mColor1}2.& $ \frac{64}{7}$ &$\frac27$ &$\frac67$ & $c_2=2$& 0&  $\frac{232}{7}$&$\frac{16}{7}$&$\frac{6}{7}$\\
\bottomrule

\end{tabular}
\caption{$(3,9)$ admissible solutions of type $W_{1,i}$ from the $(3,3)$ models.}
\label{tabW_33}
\end{table}

\begin{table}[H]  
\centering
\aboverulesep = 0pt
\belowrulesep = 0pt
\begin{tabular}{c|ccc|cc|ccc}\toprule
S. No.  & $c$& $h_1$ &  $h_2$  & $b_1$& $b_2$&  $c$&$h_1$ &$h_2$  \\  
\midrule
  1.& $\frac{48}{7}$ & $\frac17$ & $\frac57$    & $b_1 \geq 0$& $b_2=1+3b_1$&  $\frac{384}{7}$&$\frac{22}{7}$&$\frac{19}{7}$\\
\rowcolor{mColor1}2.& $\frac{64}{7}$&$\frac27$&$\frac67$& $b_1 =1+3x$& $b_2=3+4x$ ($x\in \BZ_{\geq0}$)&  $\frac{400}{7}$&$\frac{23}{7}$&$\frac{20}{7}$\\

\bottomrule

\end{tabular}
\caption{$(3,9)$ admissible solutions of type $V_0$ using $(3,3)$ solutions.}
\label{tabV_33}
\end{table}

\section{Discussion of results}

\subsection{GHM duality for VVMFs}

Let $\mathbb{X}$ be a known RCFT with $n$ characters and $\widetilde{\mathbb{X}}$ be its GHM dual\cite{Gaberdiel:2016zke}. One then has
\begin{equation}
\widetilde{\mathbb{X}}^T\cdot M \cdot \mathbb{X} = J(\tau) + b\ ,
\end{equation}
where $b$ is a constant and $M$ is the diagonal matrix of multiplicities, $M=\text{Diag}(m_0=1,m_1,\ldots,m_{n-1}$. Let $\Xi$ (resp. $\widetilde{\Xi}$) denote the matrix of the basis obtained using the methods of Bantay and Gannon starting from $\mathbb{X}$ (resp. $\widetilde{\mathbb{X}}$). Experimentally using the examples of the $B_{r,1}$ RCFTs and their GHM duals we observe that
\begin{equation}
\widetilde{\Xi}(\tau)^T\cdot M \cdot \Xi(\tau) = J(\tau)\ M + B\ ,
\end{equation}
where $B$ is a constant matrix. This formula determines the $\mathcal{Y}$-matrix for $\widetilde{\mathbb{X}}$  given the $\mathcal{Y}$-matrix for $\mathbb{X}$.

\subsection{RCFTs for $(3,6\ell)$ admissible solutions}

We have found a large number of admissible solutions using the quasi-characters given by the methods of Bantay and Gannon. While, our focus was largely on $(3,6)$ theories, we obtained larger values of Wronskian index as well. Two questions arise from our work.
\begin{enumerate}
\item Which among these admissible solutions correspond to RCFTs?
\item Do all $(3,6)$ admissible solutions arise as the  admissible solutions that we obtain starting from $(3,0)$ theories? In other words, are there any $(3,6)$ admissible solutions that do not arise from $(3,0)$ theories.  The analogous statement for the two-character case has been proven\cite{Das:2023qns}. 
\end{enumerate}
We answer the first question in the sequel.

\subsubsection*{A simple class of $(3,6)$ and $(3,12)$ RCFTs}

Let $\mathbb{X}$ denote the VVMF associated with a $(3,0)$ RCFT. Then, $\mathbb{X}^{[1]}:=J^{1/3}\,\mathbb{X}$ is the VVMF associated with the $(3,6)$ obtained by tensoring the RCFT with the single character $E_{8,1}$ RCFT. Similarly, $\mathbb{X}^{[2]}:=J^{2/3}\mathbb{X}$ is the VVMF associated with the $(3,12)$ obtained by tensoring the RCFT with the single character $E_{8,1}^{\otimes2}$ RCFT or the $D_{16,1}$ single character RCFT. These RCFT's share the same $S$-matrix as the original RCFT but the $T$-matrix is multiplied by $\omega=\exp(- 2\pi i/3)$ for the RCFT with character $\mathbb{X}^{[1]}$ and $\omega^2$ for the RCFT with character $\mathbb{X}^{[2]}$.

For the $B_{r,1}$ theories, starting from a $B_{r+8}$ theory, we obtain the $E_{8,1}\otimes B_{r,1}$ theory as one of the admissible theories of type $U_{1,2}$ when $b_2=2^r$ (see table). 

Starting from the GHM dual of the $B_{r,1}$ theory, $\widetilde{B_{r,1}}$, with $c=\frac{47}2 -r$ for $r\in [0,8]$,  we obtain the $E_{8,1}\otimes \widetilde{B_{8-r,1}}$ theory as one of the admissible theories of type $U_{1,2}$.  For $r\in[0,23]$, we also recover the the $(3,0)$ $B_{r,1}$ theory as the type $U_{1,1}$ admissible solution of the $\widetilde{B_{23-r,1}}$ RCFT. The two RCFTs are connected by a set of admissible $(3,6)$ solutions.

Similarly, for $r\in [1,15]$ and $r\neq 0\mod 4$, we obtain the $(3,0)$ $D_{r,1}$ RCFT as the end-point of a series of $(3,6)$ admissible solutions of type $U_{1,1}$ starting from the $\widetilde{D_{24-r,1}}$ RCFT. We also obtain the $(3,6)$ RCFT $E_{8,1}\otimes D_r$ as a type $U_{1,2}$ admissible solution starting from the $D_{r+8,1}$ RCFT. A similar statement holds for the GHM dual of the $D_r$ RCFTs.

Characters $\mathbb{X}^{[2]}$ arise as $(3,12)$ RCFTs of type $U_2$. For instance,  when $\mathbb{X}$ is the character associated with $B_{16,1}$, $\mathbb{X}^{[2]}$ is obtained as a type $U_2$ admissible solution of the $\widetilde{B_{7,1}}$ RCFT when $b_1=b_2=1$.

\subsubsection*{Three character RCFT's with 3 or 4 primaries}

When is an admissible character associated with an RCFT? The admissible characters that we have constructed all come with known S and T matrices. These are the modular data. Further, the Verlinde coefficients are the same as the original RCFT, they are guaranteed to be positive. The additional data is in the form the fusion matrices $F$ and braiding matrices $B$\cite{Moore:1988qv}.   

In the examples that we have considered, the $B_r$ theories correspond to RCFT's with 3 primaries. Further, the $D_r$ theories ($r\neq 0\mod 4$) correspond to RCFT's with 4 primaries. Rayhaun has classified all RCFT's with 3 or 4 primaries and central charge $c<24$. He obtains all of them as admissible solutions of type $U$\cite[see Appendices E3,E4]{Rayhaun:2023pgc}. This is consistent with our results.

\subsubsection*{$(3,6)$ RCFTs from generalized cosets}

Let $(c,h_1,h_2,\ldots,h_{n-1})$ be the data of a known RCFT with $n$ characters and  Wronskian index $\ell$. Let 
$(\tilde{c},\tilde{h}_1,\tilde{h}_2,\ldots,\tilde{h}_{n-1})$ be the data of the generalized coset dual RCFT with Wronskian index $\tilde{\ell}$. Then, one has the relations
\begin{align*}
c+\tilde{c} & = N \ ,\\
h_i +\tilde{h}_i &= m_i \quad \text{with } m_i\geq 	2\ ,\\
\ell + \tilde{\ell} & = n^2 + (2N-1)n -6 (m_1+\ldots m_{n-1}) \ ,
\end{align*}
where $N=0\mod8$ is the central charge of a single-character RCFT. We will be interested in the coset duals of $(3,0)$ RCFTs which have Wronskian index $\tilde{\ell}=6$. Following Mukhi, Poddar, Singh, we look for cases where $m_1=m_2=2$. We then obtain $(3,6)$ theories when $N=32$, i.e., these theories are generalized cosets of the single character CFTs with $c=32$.

In unpublished work\cite{DGS}, Das et al. constructed $(3,6)$ RCFT's as a generalized coset of the  32 dimensional complete Kervaire lattices by one of the $(3,0)$ theories appearing from affine Kac-Moody Lie algebras. These RCFTs appear in our constructions as type $W$ and type $U$ admissible characters for particular values of the parameters. Tables \ref{RCFTW} and \ref{RCFTU} provide the relevant details of how these RCFTs are obtained in our construction of $(3,6)$ admissible solutions. 
\begin{table}[h]
\begin{center}
\begin{tabular}{c|c |c|l }
$(3,0)$ CFT &  No.CFTs  & $(3,6)$ CFT &\hspace*{0.5in} $b$ values that give an RCFT\\ \hline
\multirow{3}{*}{$A_{1,1}^2\sim D_{2,1}$ $c=2$} &   \multirow{3}{*}{47} &\multirow{3}{*}{$W_0(D_{6,1})$}&$90,106,122,138^3,154^2,170^2,186^4,202,$ \\
&&&$218^4,234,250^3,266^3,282^2,314^4,330,$\\
&&&$346^2,362,378^2,442,474^2,506,510,570$\\ \hline
$A_{3,1}\sim D_{3,1}$ $c=3$ & 15 & \multirow{2}{*}{$W_0(D_{5,1})$} &$129,161,201,225^2,273,317,$\\ &&& $321,337,353^2,377,497,521$\\ \hline
$A_{2,1}^2$ &9 & $W_0(A_{2,1}^2,c=4)$& $112,166,184,220,316,328^2,352,448$\\ \hline
$A_{4,1}$ & 4 & $W_0(A_{4,1},c=4)$&$168,288,368,552$\\ \hline
$D_{5,1}$ &6 & $W_0(D_{3,1})$ & $195,291,371,419,451,531$\\ \hline
$D_{7,1}$ & 3 & $W_0(D_{1,1})$ & $181,277,469$\\ \hline
\end{tabular}
\end{center}
\caption{$(3,6)$ RCFTs that are obtained from type $W_0$ admissible characters of $(3,0)$ theories. A superscript indicates the number of distinct RCFTs for that value of $b$.}\label{RCFTW}
\end{table}

\begin{table}[h]
\begin{center}
\begin{tabular}{c|c |c|l }
$(3,0)$ CFT & No. CFTs  & $(3,6)$ CFT & $b_i$ values that give an RCFT\\ \hline
$D_{9,1}$ &4 & $U_{1,2}(\widetilde{D_{1,1}})$ & $b_2=64,72,100,120$\\ \hline
$D_{10,1}$ &11 & $U_{1,2}(\widetilde{D_{2,1}})$ & $b_2=8,12,16,20^2,24,32^2,40,56^2$\\ \hline
$D_{14,1}$ & 2 & $U_{1,2}(\widetilde{D_{6,1}})$ & $b_2=1^2$ \\ \hline
$E_{6,1}^2$ & 3 & $U_{1,2}(A_{2,1}^2,c=20)$ & $b_2=9,15,27$ \\ \hline
$E_{7,1}^2$ & 6 & $U_{1,1}(\widetilde{D_{6,1}})$ &$b_1=4,6,8,12,20,36$ \\ \hline 
\end{tabular}
\end{center}
\caption{$(3,6)$ RCFTs that are obtained from type $U_{1,i}$ admissible characters of $(3,0)$ theories. A superscript indicates the number of distinct RCFTs for that value of $b$.}\label{RCFTU}
\end{table}

\section{Concluding Remarks}

This paper largely focused on constructing quasi-characters associated with $(3,0)$ and $(3,3)$ RCFTs. The focus was to construct $(3,6)$ and $(3,9)$ admissible solutions using these quasi-characters. A nice geometric feature emerges -- all admissible characters arise as integral points that lie on the boundary and interior of a polytope. This observation is not restricted to the three-characters examples considered in this paper but is valid for higher numbers of characters as well.

A slight generalization of our construction is to permit rational points in the polytope. The rational number corresponds to a denominator that is fixed to a particular value. This will require rescaling of the non-identity characters to regain non-negative integrality of the $q$-series. This will also modify the $S$-matrix. We will illustrate this with a particular example. Consider the $(3, 0)$ RCFT with $c = 20,\ (h_1, h_2) = (\frac{4}{3}, \frac{5}{3} ),\ (m_1,m_2)=(4,4)$. In this example, consider the admissible characters of type $U_2$. 
\begin{equation}
    U_2=\begin{pmatrix}
   \frac{1}{q^{5/6}} +4 \left(27 b_1+3 b_2+20\right) q^{1/6}+\left(9936 b_1+132 b_2+46790\right) q^{7/6}+\cdots \\
\frac{b_1}{q^{1/2}}+3 \left(52 b_1-6 b_2+405\right) q^{\frac{1}{2}}+18 \left(419 b_1-6 b_2+11340\right) q^{3/2}+\cdots \\
\frac{b_2}{q^{1/6}}+\left(-1458 b_1+16 b_2+8748\right) q^{5/6}+\left(-40824 b_1+98 b_2+776385\right) q^{11/6}+\cdots 
\end{pmatrix}
\end{equation}
Focusing on the identity character, we see that choosing $108b_1\in \mathbb{Z}$ and $12b_2\in \mathbb{Z}$ does not change the integrality of the identity character to \hl{100} orders that we checked. However,  if the degeneracies associated with row $2$ and row $3$ are scaled by $n_1$ and $n_2$ to regain integrality, then the multiplicities will be rescaled to $\frac{m_1}{n_1^2}$ and $\frac{m_2}{n_2^2}$, so that the torus partition function continues to be modular invariant.  Thus, we can choose $2b_1\in \mathbb{Z}$ and $2b_2\in\mathbb{Z}$. Interestingly, it would change the $U_{1,2}$ upper bound to $b_{2}\leq \frac{135}{2}$. Consider  $U_{1,2}$ at $b_{2}= \frac{135}{2}$. We then multiply row $2$ and row $3$ by $2$ to regain integrality. This gives rise to admissible solutions but the S-matrix is modified with $(m_1,m_2)=(1,1)$. This $S$-matrix does not lead to good fusion rules and thus, we do not obtain sensible solution for fractional values of $b_i$.
For $b_1=0$ and $b_2=\frac{135}{2}$ we get a $(3, 0)$ admissible solution with $c = 20,\ (h_1, h_2) = (\frac{7}{3}, \frac{2}{3} )$. Consistent with our observation, this was already shown to not be an RCFT\cite{Das:2022uoe}. So it appears that fractional values are not permitted in the examples that we considered.

Using a duality due to Bantay and Gannon, we were able to obtain the S-matrices and the fusion matrices for the 15 known  $(3,3)$ admissible characters. This enables us to rule out 8 of the 15 arising from the MLDE method. The remaining 7 all have the same $S$-matrix and have three primaries. This rank-three $S$-matrix appears in the classification of Modular Tensor Categories by Rowell et al.\cite[see Theorem 3.2]{Rowell:2009}. 

A natural extension of our work is carry out similar constructions for theories with four and five characters. For four character theories, the Wronskian index is even. The Bantay-Gannon dual of $(4,0)$ theories have Wronskian index $4$ and $(4,2)$ theories get mapped to $(4,2)$ theories. It appears that it might suffice to obtain all $(4,0)$ and $(4,2)$ theories and use our method of quasi-characters and Bantay-Gannon duality to obtain theories with higher values of Wronskian index. This is being currently studied\cite{SG-AS}.\\

\noindent \textbf{Acknowledgments:} JS would like to thank Chethan N. Gowdigere, Sachin Kala, Arpan Bhattacharyya, Saptaswa Ghosh, and Sounak Pal for useful discussions. JS is supported by the project SPON/SERB/63021. AS is supported by a CSIR Senior Research Fellowship.


\appendix

\section{Modular background}

In this appendix we define some modular forms of $PSL(2,\mathbb{Z})$ that are used in this paper. This is to provide a quick description to set up notations. For details, refer to the introduction to the topic by Zagier\cite{Zagier:1992,Zagier:2008}. 

A modular form of weight $w$  is a holomorphic function on the Poincar\'e upper half plane $\mathbb{H}$ such that
\begin{equation}
\phi\left(\tfrac{a\tau+b}{c\tau+d}\right) = (c\tau+d)^w\  
\phi(\tau)
\end{equation}
for $\left(\begin{smallmatrix} a & b \\  c & d \end{smallmatrix}\right)\in PSL(2,\BZ)$. The invariance under $\tau\rightarrow \tau+1$ implies that one has the expansion
\[
\phi(\tau) = \sum_{n=-\infty}^{\infty} a_n q^n\ .
\]
If the modular form is bounded as $\tau \rightarrow i\infty$ (or $q\rightarrow0$), one has $a_n=0$ for $n<0$. Such modular forms are called \textit{holomorphic} modular forms. A modular form is called weakly holomorphic if $a_n=0$ for some $n<-N$ where $N$ is a positive integer. 

 We  provide explicit formulae for the Eisenstein series $E_2, E_4,$ and $E_6$, the cusp form of weight  $12$, and Klein-$J$ invariant. We define $q=\exp(2\pi i \tau)$. 
\begin{align*}
E_2(\tau) &= 1 - 24\sum\limits_{n=1}^{\infty} \frac{n\, q^n}{(1-q^n)} = 1-24q-72q^2-96q^3  +\ldots \\
E_4(\tau) &= 1 + 240\sum\limits_{n=1}^{\infty} \frac{n^3\, q^n}{(1-q^n)} = 1 + 240 q + 2160 q^2 + 6720 q^3  + \ldots \\
E_6(\tau) &= 1 - 504\sum\limits_{n=1}^{\infty} \frac{n^5\, q^n}{(1-q^n)} = 1 - 504 q - 16632 q^2 - 122976 q^3 +\ldots \\
\Delta(\tau) &=   \frac{E_4^3 - E_6^2}{1728} = q - 24 q^2 + 252 q^3 - 1472 q^4 + 4830 q^5 + \ldots \\
J(\tau) &= \frac{E_4^3}{\Delta} = \frac{1}{q} + 744 + 196884 q + 21493760 q^2 +\ldots
\end{align*}
Further, only the $J$-invariant is a weakly holomorphic modular function while the others are holomorphic modular forms. 

\noindent The Serre-Ramanujan covariant derivative, acting on a modular form of weight $w$,  is defined by
\bea
\mathcal{D}_w := q\frac{d}{dq} - \frac{w}{12} E_2.
\eea
It maps a modular form of weight $w$ to one of weight $(w+2)$. The operators $\nabla_i$ defined in the main text do not change the weight.

\section{Admissible $(3,6)$ Admissible Solutions obtained from $(3,0)$ theories}\label{36Admissible}

This appendix provides a complete listing of $(3,6)$ admissible solutions of type $U$, $W$ and $V$ that we obtain starting from all $(3,0)$ examples that we have considered in this paper.
\clearpage

\subsection{$(3,6)$ admissible solutions of type $U$}

\begin{longtable}{c|ccc|c|cc}\toprule

 & \multicolumn{3}{c|}{(3,0)}& &\multicolumn{2}{c}{(3,6)}\\
 \multirow{-2}{5pt}{S. No.}& $c$ &  $h_1$ & $h_2$ & $b_i$&$h_1$& $h_2$\\
  \midrule
 1.& $\frac{17}{2}$&  $\frac{1}{2}$ & $\frac{17}{16}$& $b_2\leq1$&$\frac{1}{2}$& $\frac{1}{16}$\\
 \rowcolor{mColor2} 2.& $\frac{19}{2}$&  $\frac{1}{2}$ & $\frac{19}{16}$& $b_2\leq2$&$\frac{1}{2}$& $\frac{3}{16}$\\
 3.& $\frac{21}{2}$&  $\frac{1}{2}$ & $\frac{21}{16}$& $b_2\leq5$&$\frac{1}{2}$& $\frac{5}{16}$\\
\rowcolor{mColor2} 4.& $\frac{23}{2}$&  $\frac{1}{2}$ & $\frac{23}{16}$& $b_2\leq11$&$\frac{1}{2}$& $\frac{7}{16}$\\
 5.& $\frac{25}{2}$&  $\frac{3}{2}$& $\frac{9}{16}$& $b_1<25^*$&$\frac{1}{2}$& $\frac{9}{16}$\\
 \rowcolor{mColor2} 6.& $\frac{27}{2}$&  $\frac{3}{2}$& $\frac{11}{16}$& $b_1<27^*$&$\frac{1}{2}$& $\frac{11}{16}$\\
 7.& $\frac{29}{2}$&  $\frac{3}{2}$& $\frac{13}{16}$& $b_1<29^*$&$\frac{1}{2}$& $\frac{13}{16}$\\
 \rowcolor{mColor2} 8.& $\frac{31}{2}$& $\frac{3}{2}$& $\frac{15}{16}$& $b_1<31^*$& $\frac{1}{2}$&$\frac{15}{16}$\\
 9.& $\frac{33}{2}$& $\frac{3}{2}$& $\frac{17}{16}$& $b_1<33^*$& $\frac{1}{2}$&$\frac{17}{16}$\\
 & & & & $b_2\leq1$& $\frac{3}{2}$&$\frac{1}{16}$\\
 \rowcolor{mColor2} 10.& $\frac{35}{2}$& $\frac{3}{2}$& $\frac{19}{16}$& $b_1<35^*$& $\frac{1}{2}$&$\frac{19}{16}$\\
 \rowcolor{mColor2}& & & & $b_2\leq2$& $\frac{3}{2}$&$\frac{3}{16}$\\
 11.& $\frac{37}{2}$& $\frac{3}{2}$& $\frac{21}{16}$& $b_1<37^*$& $\frac{1}{2}$&$\frac{21}{16}$\\
 & & & & $b_2\leq6$& $\frac{3}{2}$&$\frac{5}{16}$\\
 \rowcolor{mColor2}12.& $\frac{39}{2}$& $\frac{3}{2}$& $\frac{23}{16}$& $b_1<39^*$& $\frac{1}{2}$&$\frac{23}{16}$\\
 \rowcolor{mColor2}& & & & $b_2\leq17$& $\frac{3}{2}$&$\frac{7}{16}$\\
 13.& $\frac{41}{2}$& $\frac{3}{2}$& $\frac{25}{16}$& $b_1<41^*$& $\frac{1}{2}$&$\frac{25}{16}$\\
 & & & & $b_2\leq45$& $\frac{3}{2}$&$\frac{9}{16}$\\
 \rowcolor{mColor2} 14.& $\frac{43}{2}$& $\frac{3}{2}$& $\frac{27}{16}$& $b_1<43^*$& $\frac{1}{2}$&$\frac{27}{16}$\\
 \rowcolor{mColor2}& & & & $b_2\leq125$& $\frac{3}{2}$&$\frac{11}{16}$\\
 15.& $\frac{45}{2}$& $\frac{3}{2}$& $\frac{29}{16}$& $b_1<45^*$& $\frac{1}{2}$&$\frac{29}{16}$\\
 & & & & $b_2\leq398$& $\frac{3}{2}$&$\frac{13}{16}$\\
 \rowcolor{mColor2}16.& $\frac{47}{2}$& $\frac{3}{2}$& $\frac{31}{16}$& $b_1<47^*$& $\frac{1}{2}$&$\frac{31}{16}$\\
 \rowcolor{mColor2}& & & & $b_2\leq2185$& $\frac{3}{2}$&$\frac{15}{16}$\\
  \bottomrule
\caption{$(3,6)$ admissible solutions of type $U_{1,i}$ using $(3,0)$ $B_{r,1}$ theories and their GHM duals. An asterisk indicates that we obtain another $(3,0)$ admissible solution on equality.}
\label{tabU_Ising}
\end{longtable}

\newpage

\begin{longtable}{c|ccc|c|cc}\toprule

 & \multicolumn{3}{c|}{(3,0)}& &\multicolumn{2}{c}{(3,6)}\\
 \multirow{-2}{5pt}{S. No.}& $c$ &  $h_1$ & $h_2$ & $b_i$&$h_1$& $h_2$\\
  \midrule
1.& $9$&  $\frac{1}{2}$ & $\frac{9}{8}$& $b_2\leq1$&$\frac{1}{2}$& $\frac{1}{8}$\\
\rowcolor{mColor2} 2.& $11$&  $\frac{1}{2}$ & $\frac{11}{8}$& $b_2\leq5$&$\frac{1}{2}$& $\frac{3}{8}$\\3.& $13$&  $\frac{3}{2}$& $\frac{5}{8}$& $b_1<26^*$&$\frac{1}{2}$& $\frac{5}{8}$\\
\rowcolor{mColor2} 4.& $15$&  $\frac{3}{2}$& $\frac{7}{8}$& $b_1<30^*$&$\frac{1}{2}$& $\frac{7}{8}$\\
 5.& $17$& $\frac{3}{2}$& $\frac{9}{8}$& $b_1<34^*$& $\frac{1}{2}$&$\frac{9}{8}$\\
 & & & & $b_2\leq1$& $\frac{3}{2}$&$\frac{1}{8}$\\
 \rowcolor{mColor2} 6.& $19$& $\frac{3}{2}$& $\frac{11}{8}$& $b_1<38^*$& $\frac{1}{2}$&$\frac{11}{8}$\\
 \rowcolor{mColor2}& & & & $b_2\leq7$& $\frac{3}{2}$&$\frac{3}{8}$\\
 7.& $21$& $\frac{3}{2}$& $\frac{13}{8}$& $b_1<42^*$& $\frac{1}{2}$&$\frac{13}{8}$\\
 & & & & $b_2\leq53$& $\frac{3}{2}$&$\frac{5}{8}$\\
 \rowcolor{mColor2} 8.& $23$& $\frac{3}{2}$& $\frac{15}{8}$& $b_1<46^*$& $\frac{1}{2}$&$\frac{15}{8}$\\
 \rowcolor{mColor2}& & & & $b_2<575^*$& $\frac{3}{2}$&$\frac{7}{8}$\\
   \midrule
 9.& $10$& $\frac{1}{2}$& $\frac{5}{4}$& $b_2\leq2$& $\frac{1}{2}$&$\frac{1}{4}$\\
 \rowcolor{mColor2} 10.& $14$& $\frac{3}{2}$& $\frac{3}{4}$& $b_1<28^*$& $\frac{1}{2}$&$\frac{3}{4}$\\
 11.& $18$& $\frac{3}{2}$& $\frac{5}{4}$& $b_1<36^*$& $\frac{1}{2}$&$\frac{5}{4}$\\
 & & & & $b_2\leq3$& $\frac{3}{2}$&$\frac{1}{4}$\\
 \rowcolor{mColor2}12.& $22$& $\frac{3}{2}$& $\frac{7}{4}$ &  $b_1<44^*$& $\frac{1}{2}$&$\frac{7}{4}$\\
 \rowcolor{mColor2}& & & & $b_2<154^*$& $\frac{3}{2}$&$\frac{3}{4}$\\
 \bottomrule
\caption{$(3,6)$ admissible solutions of type $U_{1,i}$ using $(3,0)$ $D_{r,1}$ theories  and their GHM duals. An asterisk indicates that we obtain another $(3,0)$ admissible solution on equality.}
\label{tabU_U12}
\end{longtable}

\begin{longtable}{c|ccc|c|cc}\toprule

 & \multicolumn{3}{c|}{(3,0)}& &\multicolumn{2}{c}{(3,6)}\\
 \multirow{-2}{5pt}{S. No.}& $c$ &  $h_1$ & $h_2$ & $b_i$&$h_1$& $h_2$\\
  \midrule
 1.& $\frac{60}{7}$& $\frac{3}{7}$& $\frac{8}{7}$& $b_2\leq1$& $\frac{3}{7}$&$\frac{1}{7}$\\
 \rowcolor{mColor2}2.& $\frac{68}{7}$& $\frac{3}{7}$& $\frac{9}{7}$& $b_2\leq3$& $\frac{3}{7}$&$\frac{2}{7}$\\
 3.& $\frac{100}{7}$& $\frac{5}{7}$& $\frac{11}{7}$& $b_2\leq55$& $\frac{5}{7}$&$\frac{4}{7}$\\
 \rowcolor{mColor2}4.& $\frac{108}{7}$& $\frac{6}{7}$& $\frac{11}{7}$& $b_2<39^*$& $\frac{6}{7}$&$\frac{4}{7}$\\
 5.& $\frac{116}{7}$& $\frac{8}{7}$& $\frac{10}{7}$& $b_1\leq1$& $\frac{1}{7}$&$\frac{10}{7}$\\
 & & & & $b_2\leq16$& $\frac{8}{7}$&$\frac{3}{7}$\\
 \rowcolor{mColor2}6.& $\frac{124}{7}$& $\frac{9}{7}$& $\frac{10}{7}$& $b_1\leq4$& $\frac{2}{7}$&$\frac{10}{7}$\\
 \rowcolor{mColor2}& & & & $b_2\leq27$& $\frac{9}{7}$&$\frac{3}{7}$\\
 5.& $\frac{156}{7}$& $\frac{11}{7}$& $\frac{12}{7}$& $b_1\leq95$& $\frac{4}{7}$&$\frac{12}{7}$\\
 & & & & $b_2<130^*$& $\frac{11}{7}$&$\frac{5}{7}$\\
 \rowcolor{mColor2}6.& $\frac{164}{7}$& $\frac{11}{7}$& $\frac{13}{7}$& $b_1\leq65$& $\frac{4}{7}$&$\frac{13}{7}$\\
 \rowcolor{mColor2}& & & & $b_2\leq436$& $\frac{11}{7}$&$\frac{6}{7}$\\
  7.& $\frac{236}{7}$& $\frac{16}{7}$& $\frac{17}{7}$& $b_1=17x$& $\frac{9}{7}$&$\frac{17}{7}$\\
 & \multicolumn{3}{c|}{$x\in\mathbb{Z}_{>0}$}& $b_2=17x$& $\frac{16}{7}$&$\frac{10}{7}$\\
 \bottomrule
\caption{$(3,6)$ admissible solutions using $(3,0)$ solutions of type $U_{1,i}$ from type  $LY_2$ theories.}
\label{tabU_Ly2}
\end{longtable}

\begin{longtable}{c|ccc|c|cc}\toprule

 & \multicolumn{3}{c|}{(3,0)}& &\multicolumn{2}{c}{(3,6)}\\
 \multirow{-2}{5pt}{S. No.}& $c$ &  $h_1$ & $h_2$ & $b_i$&$h_1$& $h_2$\\
  \midrule
 1.& $\frac{44}{5}$& $\frac{2}{5}$& $\frac{6}{5}$& $b_2\leq1$& $\frac{2}{5}$&$\frac{1}{5}$\\
 \rowcolor{mColor2}2.& $\frac{52}{5}$& $\frac{3}{5}$& $\frac{6}{5}$& $b_2\leq3$& $\frac{3}{5}$&$\frac{1}{5}$\\
 3.& $\frac{68}{5}$& $\frac{4}{5}$& $\frac{7}{5}$& $b_2\leq11$& $\frac{4}{5}$&$\frac{2}{5}$\\
 \rowcolor{mColor2}4.& $\frac{76}{5}$& $\frac{4}{5}$& $\frac{8}{5}$& $b_2<57^*$& $\frac{4}{5}$&$\frac{3}{5}$\\
 5.& $\frac{84}{5}$& $\frac{6}{5}$& $\frac{7}{5}$& $b_1\leq1$& $\frac{1}{5}$&$\frac{7}{5}$\\
 & & & & $b_2\leq17$& $\frac{6}{5}$&$\frac{2}{5}$\\
 \rowcolor{mColor2}6.& $\frac{92}{5}$& $\frac{6}{5}$& $\frac{8}{5}$& $b_1\leq4$& $\frac{1}{5}$&$\frac{8}{5}$\\
 \rowcolor{mColor2}& & & & $b_2\leq59$& $\frac{6}{5}$&$\frac{3}{5}$\\
 5.& $\frac{108}{5}$& $\frac{7}{5}$& $\frac{9}{5}$& $b_1\leq16$& $\frac{2}{5}$&$\frac{9}{5}$\\
 & & & & $b_2<459^*$& $\frac{7}{5}$&$\frac{4}{5}$\\
 \rowcolor{mColor2}6.& $\frac{116}{5}$& $\frac{8}{5}$& $\frac{9}{5}$& $b_1\leq100$& $\frac{3}{5}$&$\frac{9}{5}$\\
 \rowcolor{mColor2}& & & & $b_2\leq165$& $\frac{8}{5}$&$\frac{4}{5}$\\
  7.& $\frac{164}{5}$& $\frac{12}{5}$& $\frac{11}{5}$& $b_1=11x$& $\frac{7}{5}$&$\frac{11}{5}$\\
 & \multicolumn{3}{c|}{$x\in\mathbb{Z}_{>0}$}& $b_2=11x$& $\frac{12}{5}$&$\frac{6}{5}$\\
 \bottomrule
\caption{$(3,6)$ admissible solutions of type $U_{1,i}$ using $(3,0)$ theories of type $(LY_1)^{\otimes 2}$.}
\label{tabU_LY1s}
\end{longtable}
\clearpage

\begin{longtable}{c|ccc|c|cc}\toprule

 & \multicolumn{3}{c|}{(3,0)}& &\multicolumn{2}{c}{(3,6)}\\
 \multirow{-2}{5pt}{S. No.}& $c$ &  $h_1$ & $h_2$ & $b_i$&$h_1$& $h_2$\\
 1.& $12$& $\frac{2}{3}$& $\frac{4}{3}$& $b_2\leq4$& $\frac{2}{3}$&$\frac{1}{3}$\\ 
 \rowcolor{mColor2}2.& $20$& $\frac{4}{3}$& $\frac{5}{3}$& $b_1<6^*$& $\frac{1}{3}$&$\frac{5}{3}$\\
 \rowcolor{mColor2}& & & & $b_2<67$& $\frac{4}{3}$&$\frac{2}{3}$\\
 3.& $20$& $\frac{7}{5}$& $\frac{8}{5}$& $b_1\leq12$& $\frac{2}{5}$&$\frac{8}{5}$\\
 & & & & $b_2\leq51$& $\frac{7}{5}$&$\frac{3}{5}$\\ \bottomrule

\caption{$(3,6)$ admissible solutions of type $U_{1,i}$  using $(3,0)$ theories of type $A_{2,1}^{\otimes 2}$.}
\label{tabU}
\end{longtable}

\subsection{$(3,6)$ admissible solutions of type $W$}

\begin{longtable}{c|ccc||ll||ccc}\toprule

 & \multicolumn{3}{c||}{(3,0) Sol.}&~~~~~~$c_i$&~~~~~$b$ & \multicolumn{3}{c}{$(3, 6)$ Adm.}\\
\cmidrule(lr){2-4}  
\cmidrule(lr){7-9}          
\multirow{-2}{5pt}{S. No.}  & $c$ & $h_1$& $h_2$&  &&  $c$ &$h_1$&$h_2$\\\midrule
1.& $\frac{1}{2}$& $\frac12$& $\frac{1}{16}$&  $2\leq c_1<49^*$&$b=25c_1-49$&  $\frac{49}{2}$&$\frac12$& $\frac{33}{16}$\\
 & & & & $c_2=1$& $b=274$&  $\frac{49}{2}$&$\frac52$&$\frac{1}{16}$\\
\rowcolor{mColor2}2.& $\frac{3}{2}$& $\frac{1}{2}$& $\frac{3}{16}$&  $4\leq c_1<51^*$&$b=3 (9 c_1-34 )$&  $\frac{51}{2}$&$\frac12$& $\frac{35}{16}$\\
 3.& $\frac{5}{2}$& $\frac12$& $\frac{5}{16}$& $6\leq c_1<53^*$& $b=29c_1-159$& $\frac{53}{2}$& $\frac12$&$\frac{37}{16}$\\
 & & & & $c_2=4$& $b=268$& $\frac{53}{2}$& $\frac52$&$\frac{5}{16}$\\
 \rowcolor{mColor2}4.& $\frac{7}{2}$& $\frac{1}{2}$& $\frac{7}{16}$& $8\leq c_1<55^*$& $b=31c_1-220$& $\frac{55}{2}$& $\frac12$&$\frac{39}{16}$\\
 \rowcolor{mColor2}& \multicolumn{3}{c||}{$x\in[0,2]$}& $c_2=8+7x$& $b=263 + 512 x$& $\frac{55}{2}$& $\frac52$&$\frac{7}{16}$\\
 5.& $\frac{9}{2}$& $\frac{1}{2}$& $\frac{9}{16}$& $9\leq c_1<57^*$& $b=3 (11c_1-95)$& $\frac{57}{2}$& $\frac12$&$\frac{41}{16}$\\
 & \multicolumn{3}{c||}{$x\in[0,6]$}& $c_2=10+9x$& $b=86 + 256x$& $\frac{57}{2}$& $\frac52$&$\frac{9}{16}$\\
 \rowcolor{mColor2}6.& $\frac{11}{2}$& $\frac{1}{2}$& $\frac{11}{16}$& $11\leq c_1<59^*$& $b=35c_1-354$& $\frac{59}{2}$& $\frac12$&$\frac{43}{16}$\\
 \rowcolor{mColor2}& \multicolumn{3}{c||}{$x\in[0,16]$}& $c_2=21+11x$& $b=121 + 128x$& $\frac{59}{2}$& $\frac52$&$\frac{11}{16}$\\
 7.& $\frac{13}{2}$& $\frac{1}{2}$& $\frac{13}{16}$& $12\leq c_1<61^*$& $b=37c_1-427$& $\frac{61}{2}$& $\frac12$&$\frac{45}{16}$\\
 & \multicolumn{3}{c||}{$x\in[0,55]$}& $c_2=25+13x$& $b=48 + 64 x$& $\frac{61}{2}$& $\frac52$&$\frac{13}{16}$\\
 \rowcolor{mColor2}8.& $\frac{15}{2}$& $\frac{1}{2}$& $\frac{15}{16}$& $13\leq c_1<63^*$& $b=3(13c_1-168)$& $\frac{63}{2}$& $\frac12$&$\frac{47}{16}$\\
 \rowcolor{mColor2}& \multicolumn{3}{c||}{$x\in[0,316]$}& $c_2=33+15x$& $b=27 + 32x$& $\frac{63}{2}$& $\frac52$&$\frac{15}{16}$\\
 9.& $\frac{17}{2}$& $\frac{1}{2}$& $\frac{17}{16}$& $15\leq c_1<65^*$& $b=41c_1-585$& $\frac{65}{2}$& $\frac12$&$\frac{49}{16}$\\
 & \multicolumn{3}{c|}{$x\in\mathbb{Z}_{>0}$} &  $c_2=35+17x$& $b=10+16x$& $\frac{65}{2}$& $\frac52$&$\frac{17}{16}$\\
 \rowcolor{mColor2}10.& $\frac{19}{2}$& $\frac{1}{2}$& $\frac{19}{16}$& $16\leq c_1<67^*$& $b=43c_1-670$& $\frac{67}{2}$& $\frac12$&$\frac{51}{16}$\\
 \rowcolor{mColor2}& \multicolumn{3}{c|}{$x\in\mathbb{Z}_{>0}$} & $c_2=37+19x$& $b=5+8x$& $\frac{67}{2}$& $\frac52$&$\frac{19}{16}$\\
11.& $\frac{21}{2}$& $\frac{1}{2}$& $\frac{21}{16}$&  $17\leq c_1<69^*$&$b=3 (15c_1-253)$&  $\frac{69}{2}$&$\frac12$& $\frac{53}{16}$\\
& \multicolumn{3}{c|}{$x\in\mathbb{Z}_{>0}$} &  $c_2=23+21 x$&$b=4x$&  $\frac{69}{2}$&$\frac52$& $\frac{21}{16}$\\
\rowcolor{mColor2}12.& $\frac{23}{2}$& $\frac{1}{2}$& $\frac{23}{16}$&  $19\leq c_1<71^*$&$b=47c_1-852$&  $\frac{71}{2}$&$\frac12$& $\frac{55}{16}$\\\rowcolor{mColor2}& \multicolumn{3}{c|}{$x\in\mathbb{Z}_{>0}$} &  $c_2=47+23x$&$b=1+2x$&  $\frac{71}{2}$&$\frac52$& $\frac{23}{16}$\\
\midrule
%
 13.& $\frac{25}{2}$& $\frac32$& $\frac{9}{16}$& $c_1=146+25b$& $b\geq 0$& $\frac{73}{2}$& $\frac32$&$\frac{41}{16}$\\
 \rowcolor{mColor2}14.& $\frac{27}{2}$& $\frac32$& $\frac{11}{16}$& $c_1=325+27b$& $b\geq 0$& $\frac{75}{2}$& $\frac32$&$\frac{43}{16}$\\
 \rowcolor{mColor2} & \multicolumn{3}{c|}{$x\in\mathbb{Z}_{>0}$} & $c_2=87$& $b=736$& $\frac{77}{2}$& $\frac72$&$\frac{11}{16}$\\
 15.& $\frac{29}{2}$& $\frac32$& $\frac{13}{16}$& $c_1=616+29b$& $b\geq 0$& $\frac{77}{2}$& $\frac32$&$\frac{45}{16}$\\
 \rowcolor{mColor2}16.& $\frac{31}{2}$& $\frac32$& $\frac{15}{16}$& $c_1=1027+31b$& $b\geq 0$& $\frac{79}{2}$& $\frac32$&$\frac{47}{16}$\\
 \rowcolor{mColor2}& \multicolumn{3}{c||}{$x=0,1$}& $c_2=930+3875x$& $b=1754 + 8704x$& $\frac{79}{2}$& $\frac72$&$\frac{15}{16}$\\
 17.& $\frac{33}{2}$& $\frac32$& $\frac{17}{16}$& $c_1=1566+33b$& $b\geq 0$& $\frac{81}{2}$& $\frac32$&$\frac{49}{16}$\\
 & \multicolumn{3}{c|}{$x\in\mathbb{Z}_{>0}$} & $c_2=2024+4301x$& $b=1545 + 3840 x$& $\frac{81}{2}$& $\frac72$&$\frac{17}{16}$\\
 \rowcolor{mColor2}18.& $\frac{35}{2}$& $\frac32$& $\frac{19}{16}$& $c_1=2241+35b$& $b\geq 0$& $\frac{83}{2}$& $\frac32$&$\frac{51}{16}$\\
 \rowcolor{mColor2}& \multicolumn{3}{c|}{$x\in\mathbb{Z}_{>0}$} & $c_2=588+4655x$& $b=4 + 1664 x$& $\frac{83}{2}$& $\frac72$&$\frac{19}{16}$\\
 19.& $\frac{37}{2}$& $\frac32$& $\frac{21}{16}$& $c_1=3060+37b$& $b\geq 0$& $\frac{85}{2}$& $\frac32$&$\frac{53}{16}$\\
 & \multicolumn{3}{c|}{$x\in\mathbb{Z}_{>0}$} & $c_2=2109+4921x$& $b=139 + 704 x$& $\frac{85}{2}$& $\frac72$&$\frac{21}{16}$\\
 \rowcolor{mColor2}20.& $\frac{39}{2}$& $\frac32$& $\frac{23}{16}$& $c_1=4031+39b$& $b\geq 0$& $\frac{87}{2}$& $\frac32$&$\frac{55}{16}$\\
 \rowcolor{mColor2}& \multicolumn{3}{c|}{$x\in\mathbb{Z}_{>0}$} & $c_2=6188+1105x$& $b=222+288x$& $\frac{87}{2}$& $\frac72$&$\frac{23}{16}$\\
 21.& $\frac{41}{2}$& $\frac32$& $\frac{25}{16}$& $c_1=5162+41b$& $b\geq 0$& $\frac{89}{2}$& $\frac32$&$\frac{57}{16}$\\
 & \multicolumn{3}{c|}{$x\in\mathbb{Z}_{>0}$} & $c_2=5986 + 5125x$& $b=29+112x$& $\frac{89}{2}$& $\frac72$&$\frac{25}{16}$\\
 \rowcolor{mColor2}22.& $\frac{43}{2}$& $\frac32$& $\frac{27}{16}$& $c_1=6461+43b$& $b\geq 0$& $\frac{91}{2}$& $\frac32$&$\frac{59}{16}$\\
 \rowcolor{mColor2}& \multicolumn{3}{c|}{$x\in\mathbb{Z}_{>0}$} & $c_2=11180+5031x$& $b=8+40x$& $\frac{91}{2}$& $\frac72$&$\frac{27}{16}$\\
 23.& $\frac{45}{2}$& $\frac32$& $\frac{29}{16}$& $c_1=7936+45b$& $b\geq 0$& $\frac{93}{2}$& $\frac32$&$\frac{61}{16}$\\
 & \multicolumn{3}{c|}{$x\in\mathbb{Z}_{>0}$} & $c_2=27027+1595x$& $b=3+4x$& $\frac{93}{2}$& $\frac72$&$\frac{21}{16}$\\
 \rowcolor{mColor2}24.& $\frac{47}{2}$& $\frac32$& $\frac{31}{16}$& $c_1=9595+47b$& $b\geq 0$& $\frac{95}{2}$& $\frac32$&$\frac{63}{16}$\\
 \rowcolor{mColor2}& \multicolumn{3}{c|}{$x\in\mathbb{Z}_{>0}$} & $c_2=115197+4371x$& $b=2x$& $\frac{95}{2}$& $\frac72$&$\frac{23}{16}$\\
 \bottomrule


\caption{$(3,6)$ admissible solutions of type $W_{1,i}$ using $(3,0)$ solutions of type $B_{r,1}$ and their GHM duals}
\label{tabW_Ising-2}
\end{longtable}
\clearpage

\begin{longtable}{c|ccc||ll|ccc}\toprule

 & \multicolumn{3}{c||}{(3,0) Sol.}&   ~~~~~ $c_i$& ~~~~~$b$& \multicolumn{3}{c}{$(3, 6)$ Adm.}\\
\cmidrule(lr){2-4} 
\cmidrule(lr){7-9}          
\multirow{-2}{5pt}{S. No.}  & $c$ & $h_1$& $h_2$&  & &  $c$ &$h_1$&$h_2$\\\midrule
1.& $1$& $\frac12$& $\frac{1}{8}$&  $3\leq c_1<50^*$& $b=26c_1-75$&  $25$&$\frac12$& $\frac{17}{8}$\\
 & & & & $c_2=1$&  $b=273$&  $25$&$\frac52$&$\frac{1}{8}$\\
\rowcolor{mColor2}2.& $3$& $\frac12$& $\frac{3}{8}$&  $7\leq c_1<54^*$& $b=3(10c_1-63)$&  $27$&$\frac12$& $\frac{19}{8}$\\
 \rowcolor{mColor2}& \multicolumn{3}{c||}{$x\in[0,2]$}& $c_2=3+3x$& $b=95 + 512 x$& $27$& $\frac52$&$\frac{3}{8}$\\
 3.& $5$& $\frac12$& $\frac{5}{8}$& $10\leq c_1<58^*$&  $b=34c_1-319$& $29$& $\frac12$&$\frac{21}{8}$\\
 & \multicolumn{3}{c||}{$x\in[0,14]$}& $c_2=11+5x$&  $b=125 + 128x$& $29$& $\frac52$&$\frac{5}{8}$\\
 \rowcolor{mColor2}4.& $7$& $\frac12$& $\frac{7}{8}$& $13\leq c_1<62^*$&  $b=38c_1-465$& $31$& $\frac12$&$\frac{23}{8}$\\
 \rowcolor{mColor2}& \multicolumn{3}{c||}{$x\in[0,163]$}& $c_2=15+7x$&  $b=11+32x$& $31$& $\frac52$&$\frac{7}{8}$\\
 5.& $9$& $\frac12$& $\frac{9}{8}$& $15\leq c_1<66^*$&  $b=3(14c_1-209)$& $33$& $\frac12$&$\frac{25}{8}$\\
 & \multicolumn{3}{c|}{$x\in\mathbb{Z}_{>0}$} & $c_2=19+9x$&  $b=1+8x$& $33$& $\frac52$&$\frac{9}{8}$\\
 \rowcolor{mColor2}6.& $11$& $\frac12$& $\frac{11}{8}$& $18\leq c_1<70^*$&  $b=23(2c_1-35)$& $35$& $\frac12$&$\frac{27}{8}$\\
 \rowcolor{mColor2}& \multicolumn{3}{c|}{$x\in\mathbb{Z}_{>0}$} & $c_2=23+11x$&  $b=1+2x$& $35$& $\frac52$&$\frac{11}{8}$\\
 \midrule
 7.& $13$& $\frac32$& $\frac{5}{8}$& $c_1=222+26b$&  $b\geq0$& $37$& $\frac32$&$\frac{21}{8}$\\
 \rowcolor{mColor2}8.& $15$& $\frac32$& $\frac{7}{8}$& $c_1=806+30b$&  $b\geq0$& $39$& $\frac32$&$\frac{23}{8}$\\
 \rowcolor{mColor2}& \multicolumn{3}{c||}{$x\in[0,3]$}& $c_2=78+455x$&  $b=15+2304x$& $39$& $\frac72$&$\frac{7}{8}$\\
 9.& $17$& $\frac32$& $\frac{9}{8}$& $c_1=1886+34b$&  $b\geq0$& $41$& $\frac32$&$\frac{25}{8}$\\
 & \multicolumn{3}{c|}{$x\in\mathbb{Z}_{>0}$} & $c_2=748+561x$&  $b=365+448x$& $41$& $\frac72$&$\frac{9}{8}$\\
 \rowcolor{mColor2}10.& $19$& $\frac32$& $\frac{11}{8}$& $c_1=3526+38b$&  $b\geq0$& $43$& $\frac32$&$\frac{27}{8}$\\
 \rowcolor{mColor2}& \multicolumn{3}{c|}{$x\in\mathbb{Z}_{>0}$} & $c_2=1596+627x$&  $b=59+80x$& $43$& $\frac72$&$\frac{11}{8}$\\
11.& $21$& $\frac32$& $\frac{13}{8}$&  $c_1=5790+42b$& $b\geq0$&  $45$&$\frac32$& $\frac{29}{8}$\\
& \multicolumn{3}{c|}{$x\in\mathbb{Z}_{>0}$} &  $c_2=5292+637x$& $9+12x$&  $45$&$\frac72$& $\frac{13}{8}$\\
\rowcolor{mColor2}12.& $23$& $\frac32$& $\frac{15}{8}$&  $c_1=8742+46b$& $b\geq0$&  $47$&$\frac32$& $\frac{31}{8}$\\\rowcolor{mColor2}& & & &  $c_2=33511+575b$& $b\geq0$&  $47$&$\frac72$& $\frac{15}{8}$\\

\caption{$(3,6)$ admissible solutions of type $W_{1,i}$ using $(3,0)$ solutions of type $B_{r,1}$ (odd $r$) and their GHM duals}
\label{tabW_U12}
\end{longtable}

\clearpage
\begin{longtable}{c|ccc||ll||lcc}\toprule

 & \multicolumn{3}{c||}{(3,0) Sol.}& ~~~~~~   $c_i$&~~~~~$b$& \multicolumn{3}{c}{$(3, 6)$ Adm.}\\
\cmidrule(lr){2-4} 
\cmidrule(lr){7-9}          
\multirow{-2}{5pt}{S. No.}  & $c$ & $h_1$& $h_2$&  &&  $c$ &$h_1$&$h_2$\\\midrule
1.& $2$& $\frac12$& $\frac{1}{4}$&  $5\leq c_1<52^*$&$b=2(14c_1-65)$&  $26$&$\frac12$& $\frac{9}{4}$\\
 & & & & $c_2=2,3$& $b=2(256c_2-377)$&  $26$&$\frac52$&$\frac{1}{4}$\\
\rowcolor{mColor2}2.& $6$& $\frac12$& $\frac{3}{4}$&  $11\leq c_1<60^*$&$b=6(6c_1-65)$&  $30$
&$\frac12$& $\frac{11}{4}$\\
 \rowcolor{mColor2}& \multicolumn{3}{c||}{$x \in [0,86]$}& $c_2=10+3x$& $b=15 + 48 x$& $30$& $\frac52$&$\frac{3}{4}$\\
 3.& $10$& $\frac12$& $\frac{5}{4}$& $17\leq c_1<68^*$& $b=2(22c_1-357)$& $34$& $\frac12$&$\frac{13}{4}$\\
 & \multicolumn{3}{c||}{$x\in[0,\infty)$} & $c_2=17+5x$& $b=2x$& $34$& $\frac52$&$\frac{5}{4}$\\
 \midrule
 \rowcolor{mColor2}4.& $14$& $\frac32$& $\frac{3}{4}$& $c_1=456+28b$& $b\geq 0$& $38$& $\frac32$&$\frac{11}{4}$\\
 \rowcolor{mColor2}& \multicolumn{3}{c||}{$x\in[0,7]$}& $c_2=49+49x$& $b=146 + 640 x$& $38$& $\frac72$&$\frac{3}{4}$\\
 5.& $18$& $\frac32$& $\frac{5}{4}$& $c_1=2632+36b$& $b\geq 0$& $42$& $\frac32$&$\frac{13}{4}$\\
 & \multicolumn{3}{c||}{$x\in[0,\infty)$}& $c_2=591+25x$& $b=6+8x$& $42$& $\frac72$&$\frac{5}{4}$\\
 \rowcolor{mColor2}6.& $22$& $\frac32$& $\frac{7}{4}$& $c_1=7176+44b$& $b\geq 0$& $46$& $\frac32$&$\frac{15}{4}$\\
 \rowcolor{mColor2}& & & & $c_2=11132+154b$& $b\geq 0$& $46$& $\frac72$&$\frac{7}{4}$\\
\bottomrule

\caption{$(3,6)$ admissible solutions of type $W_{1,i}$ using $(3,0)$ solutions of type $B_{r,1}$ (even $r$) and their GHM duals}
\label{tabW_A11s}
\end{longtable}
\begin{longtable}{c|ccc||l||ccc}\toprule

 & \multicolumn{3}{c|}{(3,0) Sol.}&   ~~~~~~~~~ $c_i$, $b$& \multicolumn{3}{c}{$(3, 6)$ Adm.}\\
\cmidrule(lr){2-4} 
 
\cmidrule(lr){6-8}          
\multirow{-2}{5pt}{S. No.}  & $c$ & $h_1$& $h_2$&    &  $c$ &$h_1$&$h_2$\\\midrule
1.& $\frac{4}{7}$& $\frac{1}{7}$& $\frac{3}{7}$&  $c_1=1$, $b=326$&  $\frac{172}{7}$&$\frac17$& $\frac{17}{7}$\\
 & & & & $3\leq c_2\leq22$,  $b=55c_2-129$&  $\frac{172}{7}$&$\frac{15}7$&$\frac{3}{7}$\\
\rowcolor{mColor2}2.& $\frac{12}{7}$& $\frac{2}{7}$& $\frac{3}{7}$&  $c_1=2,4$, $b=69,696$&  $\frac{180}{7}$&$\frac27$& $\frac{17}{7}$\\
 3.& $\frac{44}{7}$& $\frac{4}{7}$& $\frac{5}{7}$& $c_1=22+44x$,  $b=71+725x$,  $x\in[0,2]$& $\frac{212}{7}$& $\frac47$&$\frac{19}{7}$\\
 & & & & $c_2=20+11x$  $b=116+138x$, $x\in [0,17]$& $\frac{212}{7}$& $\frac{18}7$&$\frac{5}{7}$\\
 \rowcolor{mColor2}4.& $\frac{52}{7}$& $\frac{4}{7}$& $\frac{6}{7}$& $c_1=33,72$, $b=528,1582$& $\frac{220}{7}$& $\frac47$&$\frac{20}{7}$\\
 \rowcolor{mColor2}& & & & $c_2=26+13x$, $b=72+76x$, $x\in [0,62]$& $\frac{220}{7}$& $\frac{18}7$&$\frac{6}{7}$\\
 5.& $\frac{60}{7}$& $\frac{3}{7}$& $\frac{8}{7}$& $c_2=13+10x$,  $b=6+9x$& $\frac{228}{7}$& $\frac{17}7$&$\frac87$\\
 \rowcolor{mColor2}6.& $\frac{68}{7}$& $\frac{3}{7}$& $\frac{9}{7}$& $c_2=17x^*$,  $b=5x$& $\frac{236}{7}$& $\frac{17}7$&$\frac97$\\
 7.& $\frac{100}{7}$& $\frac{5}{7}$& $\frac{11}{7}$& $c_2=670+55b$,  $b\geq 0$& $\frac{268}{7}$& $\frac{19}7$&$\frac{11}{7}$\\
 \rowcolor{mColor2}8.& $\frac{108}{7}$& $\frac{6}{7}$& $\frac{11}{7}$& $c_1=195+351x$ ,  $b=853 + 2299 x$, $x \in [0,3]$& $\frac{276}{7}$& $\frac67$&$\frac{25}{7}$\\
 \rowcolor{mColor2}& & & & $c_2=828+39b$,  $b\geq 0$& $\frac{276}{7}$& $\frac{20}7$&$\frac{11}{7}$\\
 9.& $\frac{116}{7}$& $\frac{8}{7}$& $\frac{10}{7}$& $c_2=1189+725x$, $b=4+44x$& $\frac{284}{7}$& $\frac{22}7$&$\frac{10}{7}$\\
 \rowcolor{mColor2}10.& $\frac{124}{7}$& $\frac{9}{7}$& $\frac{10}{7}$& $c_1=868+2108x$,  $b=41+475x$& $\frac{292}{7}$& $\frac{23}7$&$\frac{10}{7}$\\
11.& $\frac{156}{7}$& $\frac{11}{7}$& $\frac{12}{7}$&  $c_2=11310+130b$, $b\geq 0$&  $\frac{324}{7}$&$\frac{25}7$& $\frac{12}7$\\
\rowcolor{mColor2}12.& $\frac{164}{7}$& $\frac{11}{7}$& $\frac{13}{7}$&  $c_1=11070+1107x$, $b=4+17x$&  $\frac{332}{7}$&$\frac{11}{7}$& $\frac{27}{7}$\\\rowcolor{mColor2}& \multicolumn{3}{c|}{$x\in\mathbb{Z}_{>0}$} &  $c_1=30135+4797x$, $b=1+11x$&  $\frac{332}{7}$&$\frac{25}{7}$& $\frac{13}{7}$\\
 \bottomrule

\caption{$(3,6)$ admissible solutions of type $W_{1,i}$ using $(3,0)$ solutions of type $LY_2$.}
\label{tabWLY2}
\end{longtable}
\begin{longtable}{c|ccc||l||lcc}\toprule

 & \multicolumn{3}{c||}{(3,0) Sol.}&~~~~~$c_i,\ b$ & \multicolumn{3}{c}{$(3, 6)$ Adm.}\\
\cmidrule(lr){2-4}  
\cmidrule(lr){6-8}          
\multirow{-2}{5pt}{S. No.}  & $c$ & $h_1$& $h_2$&  &  $c$ &$h_1$&$h_2$\\\midrule
1.& $\frac{4}{5}$& $\frac{1}{5}$& $\frac{2}{5}$&  $c_1=1,\ b=382$&  $\frac{124}{5}$&$\frac15$& $\frac{12}{5}$\\
 & & & & $4\leq c_2\leq22,\ b=3(19c_2-62)$&  $\frac{124}{5}$&$\frac{11}{5}$&$\frac{2}{5}$\\
\rowcolor{mColor2}2.& $\frac{12}{5}$& $\frac{1}{5}$& $\frac{3}{5}$&  $c_1=3,4$, $b=3(153c_1-385)$&  $\frac{132}{5}$&$\frac15$& $\frac{13}{5}$\\
\rowcolor{mColor2} & & & & $c_2=4+3x$,$b=8+50x$,  $x \in [0,30]$& $\frac{132}{5}$& $\frac{11}{5}$&$\frac{3}{5}$\\
 3.& $\frac{28}{5}$& $\frac{2}{5}$& $\frac{4}{5}$& $c_2=18+7x$, $b=14+26x$, $x \in [0,121]$& $\frac{148}{5}$& $\frac{12}{5}$&$\frac{4}{5}$\\
 \rowcolor{mColor2}4.& $\frac{36}{5}$& $\frac{3}{5}$& $\frac{4}{5}$& $c_1=20+9x$, $b=36+154x$, $x \in [0,14]$& $\frac{156}{5}$& $\frac{3}{5}$&$\frac{14}{5}$\\
 5.& $\frac{44}{5}$& $\frac{2}{5}$& $\frac{6}{5}$& $c_1=13,24$, $b=378,1480$& $\frac{164}{5}$& $\frac{2}{5}$&$\frac{16}{5}$\\
 & \multicolumn{3}{c|}{$x\in\mathbb{Z}_{>0}$} & $c_2=11x^*$, $b=10x$& $\frac{164}{5}$& $\frac{12}{5}$&$\frac{6}{5}$\\
 \rowcolor{mColor2}6.& $\frac{52}{5}$& $\frac{3}{5}$& $\frac{6}{5}$& $c_1=26+26x$, $b=109 + 625 x$, $x \in [0,4]$& $\frac{172}{5}$& $\frac{3}{5}$&$\frac{16}{5}$\\
 \rowcolor{mColor2} & \multicolumn{3}{c|}{$x\in\mathbb{Z}_{>0}$} & $c_2=91+26x$, $b=3+7x$& $\frac{172}{5}$& $\frac{13}{5}$&$\frac{6}{5}$\\
 7.& $\frac{68}{5}$& $\frac{4}{5}$& $\frac{7}{5}$& $c_1=102+68x$, $b=49+299x$,  $x \in [0,18]$& $\frac{188}{5}$& $\frac{4}{5}$&$\frac{17}{5}$\\
  & \multicolumn{3}{c|}{$x\in\mathbb{Z}_{>0}$} & $c_2=391+119x$, $b=6+10x$& $\frac{188}{5}$& $\frac{14}{5}$&$\frac{7}{5}$\\
 \rowcolor{mColor2}8.& $\frac{76}{5}$& $\frac{4}{5}$& $\frac{8}{5}$& $c_1=171,342$, $b=2010,4510$& $\frac{196}{5}$& $\frac{4}{5}$&$\frac{18}{5}$\\
 \rowcolor{mColor2}& & & & $c_2=931+57b$, $b \geq0$& $\frac{196}{5}$& $\frac{14}{5}$&$\frac{8}{5}$\\
 9.& $\frac{84}{5}$& $\frac{6}{5}$& $\frac{7}{5}$& $c_1=253+33x$,  $b=16+25x$& $\frac{204}{5}$& $\frac{6}{5}$&$\frac{17}{5}$\\
 & \multicolumn{3}{c|}{$x\in\mathbb{Z}_{>0}$} & $c_2=1508+154x$,  $b=6+9x$& $\frac{204}{5}$& $\frac{16}{5}$&$\frac{7}{5}$\\
 \rowcolor{mColor2}10.& $\frac{92}{5}$& $\frac{6}{5}$& $\frac{8}{5}$& $c_1=1127+299x$,  $b=28+68x$& $\frac{212}{5}$& $\frac{6}{5}$&$\frac{18}{5}$\\
 \rowcolor{mColor2}& \multicolumn{3}{c|}{$x\in\mathbb{Z}_{>0}$} &  $c_2=3289 + 299x$,  $b=2+5x$&  $\frac{212}{5}$&$\frac{16}{5}$& $\frac{8}{5}$\\11.& $\frac{108}{5}$& $\frac{7}{5}$& $\frac{9}{5}$&  $c_1=3944+833x$,  $b=32+50x$&  $\frac{228}{5}$&$\frac{7}{5}$& $\frac{19}{5}$\\
& & & &  $c_2=25194+459b$,  $b \geq 0$&  $\frac{228}{5}$&$\frac{17}{5}$& $\frac{8}{5}$\\\rowcolor{mColor2}12.& $\frac{116}{5}$& $\frac{8}{5}$& $\frac{9}{5}$&  $c_1=14877+1102x$ , $b=1+11x$&  $\frac{236}{5}$&$\frac{8}{5}$& $\frac{19}{5}$\\
 \rowcolor{mColor2}& \multicolumn{3}{c|}{$x\in\mathbb{Z}_{>0}$} & $c_2=14326 + 1653 x$, $b=8+10x$& $\frac{236}{5}$& $\frac{18}{5}$&$\frac{9}{5}$\\
\bottomrule
\caption{$(3,6)$ admissible solutions of type $W_{1,i}$ using $(3,0)$ solutions of type $(LY_1)^{\otimes2}$}
\label{tabW_LY1s}
\end{longtable}

\begin{longtable}{c|ccc||ll||lcc}\toprule

 & \multicolumn{3}{c||}{(3,0) Sol.}&~~~~~$c_i$&~~~~~$b$ & \multicolumn{3}{c}{$(3, 6)$ Adml.}\\
\cmidrule(lr){2-4}  
\cmidrule(lr){7-9}          
\multirow{-2}{5pt}{S. No.}  & $c$ & $h_1$& $h_2$&  & &  $c$ &$h_1$&$h_2$\\\midrule
1.& $4$& $\frac{1}{3}$& $\frac{2}{3}$&  $3\leq c_1 \leq 7$&$b=2(135c_1-322)$&  $28$&$\frac{1}{3}$& $\frac{8}{3}$\\
 & \multicolumn{3}{c|}{$x\in\mathbb{Z}_{>0}$} & $5\leq c_1 \leq 112$&$b=(18c_2-77)$&  $28$&$\frac{7}{3}$&$\frac{2}{3}$\\
\rowcolor{mColor2}2.& $12$& $\frac{2}{3}$& $\frac{4}{3}$&  $21\leq c_1 < 162^*$&$b=24(c_1-21)$&  $36$&$\frac{2}{3}$& $\frac{10}{3}$\\
 \rowcolor{mColor2} & \multicolumn{3}{c|}{$x\in\mathbb{Z}_{>0}$} & $c_2=81+9x$& $b=2x$& $36$& $\frac{8}{3}$&$\frac{4}{3}$\\
3.& $20$& $\frac{4}{3}$& $\frac{5}{3}$& $c_2=4158 + 135 x$&$b=2x$, $x\in \mathbb{Z}_{>0}$& $44$& $\frac{10}{3}$&$\frac{5}{3}$\\ \midrule \midrule
 \rowcolor{mColor2}1.& $4$& $\frac{2}{5}$& $\frac{3}{5}$& $c_2=10+5x$& $b=106+119x$, $x\in[0,13]$& $28$& $\frac{12}{5}$&$\frac{3}{5}$\\
 2.& $20$& $\frac{7}{5}$& $\frac{8}{5}$& $c_2=625 (7 + 4 x)$& $b=49x+6$. $x\in \mathbb{Z}_{>0}$ & $44$& $\frac{17}{5}$&$\frac{8}{5}$\\
 \bottomrule

\caption{$(3,6)$ admissible sols. of type $W_{1,i}$ using $(3,0)$ solutions of type $A_{2,1}^{\otimes2}$ and $A_{4,1}$.}
\label{tabW_As}
\end{longtable}

\subsection{$(3,6)$ admissible solutions of type $V$}

\begin{longtable}{c|ccc||>{\raggedright\arraybackslash}p{0.55\linewidth}||ccc}\toprule

 & \multicolumn{3}{c||}{(3,0) Sol.}&~~~~~~~~~~ $b_1,\ b_2,\ s_i $& \multicolumn{3}{c}{$(3, 6)$ Adm.}\\
\cmidrule(lr){2-4} 
\cmidrule(lr){6-8}          
\multirow{-2}{5pt}{S. No.}  & $c$ & $h_1$& $h_2$&  &  $c$ &$h_1$&$h_2$\\\midrule
1.& $\frac{1}{2}$& $\frac12$& $\frac{1}{16}$&  $b_1 \geq 0,\ b_2=285956+1176 b_1,\ s_1=11446+49 b_1$&  $\frac{97}{2}$&$\frac32$& $\frac{65}{16}$\\
 & & & & $b_1=64756 + 102400 x$, $b_2= 63943315 + 96538624 x$, $s_2= 1001 (76 + 119 x)$&  $\frac{97}{2}$&$\frac92$&$\frac{17}{16}$\\
\rowcolor{mColor2}2.& $\frac{3}{2}$& $\frac{1}{2}$& $\frac{3}{16}$&  $b_1\geq 0,\ b_2=363627+1275 b_1,\ s_1=13497+51 b_1$&  $\frac{99}{2}$&$\frac32$& $\frac{67}{16}$\\
 \rowcolor{mColor2}& & & & $b_1=24087 + 26624 x$, $  b_2 = 26433783 + 27992064 x$,  $s_2=629 (123+133 x)$& $\frac{99}{2}$& $\frac92$&$\frac{19}{16}$\\
 3.& $\frac{5}{2}$& $\frac12$& $\frac{5}{16}$& $b_1\geq 0$, $b_2=455005+1378 b_1$, $s_1=15756+53b_1$& $\frac{101}{2}$& $\frac32$&$\frac{69}{16}$\\
 & & & & $b_1=53342 + 56320 x$, $b_2= 64559448 + 67351552 x$, $s_2= 1961 (234 + 245 x)$& $\frac{101}{2}$& $\frac92$&$\frac{21}{16}$\\
 \rowcolor{mColor2}4.& $\frac{7}{2}$& $\frac{1}{2}$& $\frac{7}{16}$& $b_1\geq 0$, $b_2=561350+1485 b_1$, $s_1=18231+55b_1$& $\frac{103}{2}$& $\frac32$&$\frac{71}{16}$\\
 5.& $\frac{9}{2}$& $\frac{1}{2}$& $\frac{9}{16}$& $b_1\geq 0$, $b_2=683970+1596 b_1$, $s_1=20930+57b_1$& $\frac{105}{2}$& $\frac32$&$\frac{73}{16}$\\
 & & & & $b_1=8 (139 + 224 x)$, $b_2= 2523377 + 3078912 x$,  $s_2= 24149 (4 + 5 x)$& $\frac{105}{2}$& $\frac92$&$\frac{25}{16}$\\
 \rowcolor{mColor2}6.& $\frac{11}{2}$& $\frac{1}{2}$& $\frac{11}{16}$& $b_1\geq 0$, $b_2=824221+1711 b_1$, $s_1=23861+59b_1$& $\frac{107}{2}$& $\frac32$&$\frac{75}{16}$\\
 7.& $\frac{13}{2}$& $\frac{1}{2}$& $\frac{13}{16}$& $b_1\geq 0$, $b_2=983507+1830 b_1$, $s_1=27032+61b_1$& $\frac{109}{2}$& $\frac32$&$\frac{77}{16}$\\
 \rowcolor{mColor2}8.& $\frac{15}{2}$& $\frac{1}{2}$& $\frac{15}{16}$& $b_1\geq 0$, $b_2=1163280+1953b_1$, $s_1=30451+63b_1$& $\frac{111}{2}$& $\frac32$&$\frac{79}{16}$\\
 \rowcolor{mColor2}& & & & $b_1=29 + 160 x$,  $b_2= 36 (64543 + 45136 x)$, $s_2= 705 (1546 + 1085 x)$& $\frac{111}{2}$& $\frac92$&$\frac{31}{16}$\\
 9.& $\frac{17}{2}$& $\frac{1}{2}$& $\frac{17}{16}$& $b_1\geq 0$, $b_2=1365040+2080b_1$, $s_1=34126+65b_1$& $\frac{113}{2}$& $\frac32$&$\frac{81}{16}$\\
 \rowcolor{mColor2}10.& $\frac{19}{2}$& $\frac{1}{2}$& $\frac{19}{16}$& $b_1\geq 0$, $b_2=1590335+2211b_1$, $s_1=38065+67b_1$& $\frac{115}{2}$& $\frac32$&$\frac{83}{16}$\\

11.& $\frac{21}{2}$& $\frac{1}{2}$& $\frac{21}{16}$&  $b_1\geq 0$, $b_2=1840761+2346b_1$, $s_1=42276+69b_1$&  $\frac{117}{2}$&$\frac32$& $\frac{85}{16}$\\
\rowcolor{mColor2}12.& $\frac{23}{2}$& $\frac{1}{2}$& $\frac{23}{16}$&  $b_1\geq 0$, $b_2=2117962+2485b_1$, $s_1=46767+71b_1$&  $\frac{119}{2}$&$\frac32$& $\frac{87}{16}$\\
\midrule
13.& $\frac{31}{2}$& $\frac32$& $\frac{15}{16}$& $b_1=4716 + 8704 x$, $b_2= 92434598 + 156088832 x$, $s_2= 761165 (55 + 93 x)$& $\frac{127}{2}$& $\frac{11}2$&$\frac{31}{16}$\\
 \bottomrule

\caption{$(3,6)$ admissible solutions of type $V_{1,i}$ using $(3,0)$ solutions of type $B_{r,1}$.}
\label{tabV_Ising}
\end{longtable}

\begin{longtable}{c|ccc||l||ccc}\toprule

 & \multicolumn{3}{c||}{(3,0) Sol.}&~~~~~~~~~~ $\hspace{2 cm}b_1,\ b_2,\ s_i$& \multicolumn{3}{c}{$(3, 6)$ Adm.}\\
\cmidrule(lr){2-4} 
\cmidrule(lr){6-8}          
\multirow{-2}{5pt}{S. No.}  & $c$ & $h_1$& $h_2$&  &  $c$ &$h_1$&$h_2$\\\midrule
1.& $1$& $\frac12$& $\frac{1}{8}$&   $b_1 \geq 0,\ b_2=323155+1225 b_1,\ s_1=12446+50b_1$&  $49$&$\frac32$& $\frac{33}{8}$\\
 & & & & $b_1=609+3584 x$,  $b_2=4 (539131 + 890240 x),\ s_2=1615 (1 + 3 x)$&  $49$&$\frac92$&$\frac{9}{8}$\\
 \rowcolor{mColor2}2.& $2$& $\frac12$& $\frac{1}{4}$&$b_1 \geq 0,\ b_2=407525+1326 b_1,\ s_1=14600+52b_1$&  $50$&$\frac32$& $\frac{17}{4}$\\
 \rowcolor{mColor2}& & & & $b_1=22+32x$,  $b_2=917721 + 35776x,\ s_2=117 (16 + x)$&  $50$&$\frac92$&$\frac{5}{4}$\\
3.& $3$& $\frac12$& $\frac{3}{8}$&  $b_1 \geq 0,\ b_2=506226+1431 b_1,\ s_1=16966+54b_1$&  $51$&$\frac32$& $\frac{35}{8}$\\
 & &&& $b_1=751+3200x$,  $b_2=27 (61505 + 152704 x),\ s_2=2907 (4 + 11 x)$& $51$& $\frac92$&$\frac{11}{8}$\\
 \rowcolor{mColor2}4.& $5$& $\frac12$& $\frac{5}{8}$& $b_1 \geq 0,\ b_2=751805+1653 b_1,\ s_1=22366+58b_1$& $53$& $\frac32$&$\frac{37}{8}$\\
 \rowcolor{mColor2}& &&& $b_1=141+160x$,  $b_2=898974 + 312736x,\ s_2=29 (1244 + 455 x)$& $53$& $\frac92$&$\frac{13}{8}$\\
 5.& $6$& $\frac12$& $\frac{3}{4}$&$b_1 \geq 0,\ b_2=901395+1770 b_1,\ s_1=25416+60b_1$&  $54$&$\frac32$& $\frac{19}{4}$\\
 & &&& $b_1=2x$,  $b_2=724275 + 5556 x,\ s_2=68013 + 539 x$& $54$& $\frac92$&$\frac{7}{4}$\\
 \rowcolor{mColor2}6.& $7$& $\frac12$& $\frac{7}{8}$& $b_1 \geq 0,\ b_2=1070740+1891 b_1,\ s_1=28710+62b_1$& $55$& $\frac32$&$\frac{39}{8}$\\
 7.& $9$& $\frac12$& $\frac{9}{8}$& $b_1 \geq 0,\ b_2=1474647+2145 b_1,\ s_1=36062+66b_1$& $57$& $\frac32$&$\frac{41}{8}$\\
  \rowcolor{mColor2}8.& $10$& $\frac12$& $\frac{5}{4}$& $b_1 \geq 0,\ b_2=1712305+2278 b_1,\ s_1=40136+68b_1$& $58$& $\frac32$&$\frac{21}{4}$\\
 9.& $11$& $\frac12$& $\frac{11}{8}$& $b_1 \geq 0,\ b_2=1975910+2415 b_1,\ s_1=44486+70b_1$& $59$& $\frac32$&$\frac{43}{8}$\\
 \midrule
 \rowcolor{mColor2}10.& $13$& $\frac32$& $\frac{5}{8}$& $b_1=1137 + 2816x$,  $b_2=6039998 + 9273088 x,\ s_2=24679 (8 + 13 x)$& $61$& $\frac{11}2$&$\frac{13}{8}$\\
 11.& $14$& $\frac32$& $\frac{3}{4}$& $b_1=106+280x$, $b_2=3213267 + 1333040 x$, $s_2= 2299 (110 + 49 x) $& $62$& $\frac{11}2$&$\frac{7}{4}$\\
\bottomrule
\caption{$(3,6)$ admissible solutions of type $V_{1,i}$ using $(3,0)$ solutions of type  $D_{r,1}$.}
\label{tabV_U12}
\end{longtable}

\begin{longtable}{c|ccc||l||ccc}\toprule

 & \multicolumn{3}{c|}{(3,0) Sol.}&    $\hspace{2 cm}b_1,\ b_2,\ s_i$& \multicolumn{3}{c}{$(3, 6)$ Adm.}\\
\cmidrule(lr){2-4} 
\cmidrule(lr){6-8}          
\multirow{-2}{5pt}{S. No.}  & $c$ & $h_1$& $h_2$&    &  $c$ &$h_1$&$h_2$\\\midrule
1.& $\frac{4}{7}$& $\frac{1}{7}$& $\frac{3}{7}$&  $b_1=8500 + 14467 x$, $b_2=9481960 + 14051712 x$, $s_1=2494 (5 + 8 x)$&  $\frac{340}{7}$&$\frac87$& $\frac{31}{7}$\\
\rowcolor{mColor2}2.& $\frac{12}{7}$& $\frac{2}{7}$& $\frac{3}{7}$&  $b_1=23+209 x$, $b_2=36 (20839 + 6479 x)$, $s_1=124 (20 + 9 x)$&  $\frac{348}{7}$&$\frac97$& $\frac{31}{7}$\\
 3.& $\frac{44}{7}$& $\frac{4}{7}$& $\frac{5}{7}$& $b_1=532 + 551 x$, $b_2=1897442 + 1088196x$, $s_1=530 (236 + 143 x)$& $\frac{380}{7}$& $\frac{11}7$&$\frac{33}{7}$\\
 & & & & $b_1=50 + 207 x$, $b_2=890018 + 533922x$, $s_1=8957 (8 + 5 x)$& $\frac{380}{7}$& $\frac{32}7$&$\frac{19}{7}$\\
 \rowcolor{mColor2}4.& $\frac{52}{7}$& $\frac{4}{7}$& $\frac{6}{7}$& $b_1=97+527 x$, $b_2=194 (6395 + 5797 x)$, $s_1=492 (97 + 99 x)$& $\frac{388}{7}$& $\frac{11}7$&$\frac{34}{7}$\\
 5.& $\frac{60}{7}$& $\frac{3}{7}$& $\frac{8}{7}$& $b_1=484+15730 x$, $b_2=2517810 + 30403230 x$, $s_2=15523 (2 + 29 x)$& $\frac{396}{7}$& $\frac{10}7$&$\frac{36}7$\\
\bottomrule
\caption{$(3,6)$ admissible solutions of type $V_{1,i}$ using $(3,0)$ solutions of type $LY_2$ and $x\in \mathbb{Z}_{>0}$.}
\label{tabV_LY2}
\end{longtable}

\begin{longtable}{c|ccc||l||ccc}\toprule

 & \multicolumn{3}{c||}{(3,0) Sol.}&~~~~~~~~~~    $b_1,\ b_2,\  s_i$& \multicolumn{3}{c}{$(3, 6)$ Adm.}\\
\cmidrule(lr){2-4} 
\cmidrule(lr){6-8}          
\multirow{-2}{5pt}{S. No.}  & $c$ & $h_1$& $h_2$&  &  $c$ &$h_1$&$h_2$\\\midrule
1.& $\frac{4}{5}$& $\frac{1}{5}$& $\frac{2}{5}$&  $b_1=2012 + 2500 x$, $b_2=2917568 + 2518750 x$, $s_1=3751 (1 + x)$&  $\frac{244}{5}$&$\frac{6}5$& $\frac{22}{5}$\\
 & & & & $b_1=15+19x$, $b_2=417449 + 21204 x$, $s_2=7409 + 434 x$&  $\frac{244}{5}$&$\frac{21}5$&$\frac{7}{5}$\\
\rowcolor{mColor2}2.& $\frac{28}{5}$& $\frac{2}{5}$& $\frac{4}{5}$&  $b_1=3788+4375 x$, $b_2=6973730 + 6845000 x$, $s_1=1887 (48 + 49 x)$&  $\frac{268}{5}$&$\frac{7}{5}$& $\frac{24}{5}$\\
\rowcolor{mColor2} & & & & $b_1=3+13x$, $b_2=706078 + 41366 x$, $s_2=222 (857 + 51 x)$& $\frac{268}{5}$& $\frac{22}{5}$&$\frac{9}{5}$\\
 3.& $\frac{36}{5}$& $\frac{3}{5}$& $\frac{4}{5}$& $b_1=6+49x$, $b_2=975090 + 107562 x$, $s_1=57681 + 7163 x$& $\frac{276}{5}$& $\frac{8}{5}$&$\frac{24}{5}$\\
 & & & & $b_1=4+50x$, $b_2=975986 + 180600 x$, $s_2=551 (142 + 27 x)$& $\frac{276}{5}$& $\frac{23}{5}$&$\frac{9}{5}$\\
 \rowcolor{mColor2}4.& $\frac{44}{5}$& $\frac{2}{5}$& $\frac{6}{5}$& $b_1=796+2204 x$, $b_2=3264955 + 4201926 x$, $s_2=1271 (32 + 49 x)$& $\frac{284}{5}$& $\frac{7}{5}$&$\frac{26}{5}$\\
 \midrule
 5.& $\frac{68}{5}$& $\frac{4}{5}$& $\frac{7}{5}$& $b_1=139 + 391 x$, $b_2=3543566 + 2183620 x$, $s_1=10434 (79 + 51 x)$& $\frac{308}{5}$& $\frac{9}{5}$&$\frac{27}{5}$\\
 \rowcolor{mColor2}6.& $\frac{76}{5}$& $\frac{4}{5}$& $\frac{8}{5}$& $b_1=2390+3750 x$, $b_2=18130982 + 23030000 x$, $s_1=188993 (7 + 9 x)$& $\frac{316}{5}$& $\frac{9}{5}$&$\frac{28}{5}$\\
 \bottomrule
\caption{$(3,6)$ admissible solutions of type $V_{1,i}$ using $(3,0)$ solutions of type $(LY_1)^{\otimes2}$ and $x\in \mathbb{Z}_{>0}$.}
\label{tabV_LY1s}
\end{longtable}

\clearpage
\begin{longtable}{c|ccc||l||ccc}\toprule

 & \multicolumn{3}{c||}{(3,0) Sol.}&~~~~~~~~~~~$\hspace{2 cm}b_1,\ b_2,\ s_i$& \multicolumn{3}{c}{$(3, 6)$ Adml.}\\
\cmidrule(lr){2-4}  
\cmidrule(lr){6-8}          
\multirow{-2}{5pt}{S. No.}  & $c$ & $h_1$& $h_2$& &  $c$ &$h_1$&$h_2$\\\midrule
1.& $4$& $\frac{1}{3}$& $\frac{2}{3}$&$b_1=6+20x$, $b_2=963517 + 26810 x$, $s_1=147 (25 + x)$&  $52$&$\frac{4}{3}$& $\frac{14}{3}$\\
 & & & &$b_1=2x$, $b_2=487942 + 3896 x$, $s_1=27144 + 225 x$&  $52$&$\frac{13}{3}$&$\frac{8}{3}$\\
\rowcolor{mColor2}2.& $12$& $\frac{2}{3}$& $\frac{4}{3}$&$b_1\geq 0$, $b_2=1971810+3384 b_1$, $s_1=83592+162b_1$&  $60$&$\frac{5}{3}$& $\frac{16}{3}$\\ \midrule\midrule
 1.& $4$& $\frac{2}{5}$& $\frac{3}{5}$& $b_1=12896+15548x$, $b_2=18921702 + 21867664 x$, $s_1=34375 (6 + 7 x)$& $52$& $\frac{7}{5}$&$\frac{23}{5}$\\
\bottomrule
\caption{
$(3,6)$ admissible solutions of type $V_{1,i}$ using $(3,0)$ solutions of type $A_{2,1}^{\otimes2}$ and $A_{4,1}$ and $x\in \mathbb{Z}_{>0}$.}\label{tabV_As}
\end{longtable}

\bibliographystyle{jhep}
\bibliography{master}  
\end{document}